\documentclass[12pt]{article}
\usepackage{amsmath}
\usepackage{amsfonts}
\usepackage{makeidx}
\usepackage{amssymb}
\usepackage{times}

\newtheorem{satz}{Theorem}[section]

\newtheorem{definition}[satz]{Definition}
\newtheorem{lemma}[satz]{Lemma}

\newtheorem{koro}[satz]{Corollary}
\newtheorem{bemerkung}[satz]{Remark}

\newtheorem{notation}[satz]{Notation}
\newenvironment{proof}{\par\noindent {\it Proof:} \hspace{7pt}}{\hfill\hbox{\vrule width 7pt depth 0pt height 7pt}
\par\vspace{10pt}}

\input{tcilatex}
\begin{document}

\title{Universal Bounds for Large Determinants from Non--Commutative H\"{o}%
lder Inequalities in Fermionic Constructive Quantum Field Theory}
\author{J.-B. Bru \and W. de Siqueira Pedra}
\date{\today }
\maketitle

\begin{abstract}
Efficiently bounding large determinants is an essential step in non--relati%
\-%
vistic fermionic constructive\ quantum field theory to prove the absolute
convergence of the perturbation expansion of correlation functions in terms
of powers of the strength $u\in \mathbb{R}$ of the interparticle
interaction. We provide, for large determinants of fermionic convariances,
\emph{sharp} bounds which hold for \emph{all} (bounded and unbounded, the
latter not being limited to semibounded) one--particle Hamiltonians. We find
the smallest \emph{universal determinant bound} to be exactly $1$. In
particular, the convergence of perturbation series at $u=0$ of any fermionic
quantum field theory is ensured if the matrix entries, with respect to some
fixed orthonormal basis, of the covariance and the interparticle interaction
decay sufficiently fast. Our proofs use H\"{o}lder inequalities for general
non--commutative $L^{p}$--spaces derived by Araki and Masuda \cite%
{Araki-Masuda}.\bigskip

\bigskip

\end{abstract}

\tableofcontents%

\section{Setup of the Problem\label{Section main results}}

The convergence of perturbation expansions in non--relativistic fermionic
constructive quantum field theory at weak coupling is ensured if the matrix
entries, with respect to some fixed orthonormal basis, of the covariance and
the interparticle interaction decay sufficiently fast and if certain
determinants arising in the expansion can be bounded efficiently. For \emph{%
any} one--particle Hamiltonian we show here how to get such bounds on
determinants from non--commutative H\"{o}lder inequalities. To our
knowledge, such estimates are unknown for the unbounded case, even for
semibounded (one--particle) Hamiltonians. The unbounded case is important,
for instance, in the context of fermionic theories in the continuum. See
also Remarks \ref{remark blabla} and \ref{remark blabla copy(1)}.

The bounds on determinants (of fermionic covariances) obtained in this way
turn out to be \emph{universal} and \emph{sharp}, in a sense to be made
precise below (cf. (\ref{universal determinant bound}) and Corollary \ref%
{theorem main super2 copy(1)}). A consequence of these estimates is that the
convergence of perturbation expansions in non--relativistic fermionic
quantum field theory is implied by decay properties of interaction and
covariance alone. Similar to \cite{dSPS}, we give bounds which do not impose
cutoffs on the Matsubara frequency, but the results obtained here are
stronger than those of \cite{dSPS} on determinants of fermionic covariances.

The paper is organized as follows: Definitions and notation are fixed in
Sections \ref{sect torus1}--\ref{sect torus2}. The problem of bounding large
determinants and the importance of our results in the context of
constructive quantum field theory are discussed in Section \ref{Sec
Constructive}. Our main results are Theorem \ref{theorem main super} and
Corollaries \ref{theorem main super2}--\ref{theorem main super2 copy(1)} of
Section \ref{Main Results}. Our approach uses H\"{o}lder inequalities for
general non--commutative $L^{p}$--spaces. See, e.g., \cite{Araki-Masuda}.
The main lines of the proofs are explained in Section \ref{Main Results},
while the technical details are postponed to Section \ref{Thecnical proofs}.

\begin{notation}
\label{remark constant}\mbox{
}\newline
A norm on a generic vector space $\mathcal{Y}$ is denoted by $\Vert \cdot
\Vert _{\mathcal{Y}}$ and the identity map of $\mathcal{Y}$ by $\mathbf{1}_{%
\mathcal{Y}}$. The space of all bounded linear operators on $(\mathcal{Y}%
,\Vert \cdot \Vert _{\mathcal{Y}})$ is denoted by $\mathcal{B}(\mathcal{Y})$%
. If $\mathcal{Y}$ is a Hilbert space, then $\left\langle \cdot ,\cdot
\right\rangle _{\mathcal{Y}}$ denotes its scalar product. Units of $C^{\ast
} $--algebras are always denoted by $\mathbf{1}$.
\end{notation}

\subsection{Spaces of Antiperiodic Functions on Discrete Tori\label{sect
torus1}}

We start by defining spaces of antiperiodic functions taking values in a
fixed Hilbert space and next give the definition of the antiperiodic
discrete delta function: \medskip

\noindent \underline{(i):} Fix $\beta \in \mathbb{R}^{+}$, an \emph{even}
integer $n\in 2${$\mathbb{N}$} and let
\begin{equation}
\mathbb{T}_{n}\doteq \left\{ -\beta +kn^{-1}\beta :k\in \left\{ 1,2,\ldots
,2n\right\} \right\} \subset \left( -\beta ,\beta \right]
\label{definition torus}
\end{equation}%
be the discrete torus of length $2\beta $. This means that $-\beta \equiv
\beta $. Pick any Hilbert space $\mathfrak{h}$ and let $\ell _{\mathrm{ap}%
}^{2}(\mathbb{T}_{n};\mathfrak{h})$ be the Hilbert space of functions from $%
\mathbb{T}_{n}$ to $\mathfrak{h}$ which are \emph{antiperiodic}. That is
here, for any $f\in \ell _{\mathrm{ap}}^{2}(\mathbb{T}_{n};\mathfrak{h})$,%
\begin{equation*}
f\left( \alpha +\beta \right) =-f\left( \alpha \right) \ ,\qquad \alpha \in
\mathbb{T}_{n}\ .
\end{equation*}%
The scalar product on $\ell _{\mathrm{ap}}^{2}(\mathbb{T}_{n};\mathfrak{h})$
is then defined to be%
\begin{equation*}
\left\langle f_{1},f_{2}\right\rangle _{\ell _{\mathrm{ap}}^{2}(\mathbb{T}%
_{n};\mathfrak{h})}\doteq n^{-1}\beta \sum\limits_{\alpha \in \mathbb{T}%
_{n}}\left\langle f_{1}\left( \alpha \right) ,f_{2}\left( \alpha \right)
\right\rangle _{\mathfrak{h}}\ ,\qquad f_{1},f_{2}\in \ell _{\mathrm{ap}%
}^{2}(\mathbb{T}_{n};\mathfrak{h})\ .
\end{equation*}

The parameter $\beta $ is interpreted as being the inverse temperature in
(fermionic and non--relativistic) quantum field theory, while $\mathfrak{h}$
refers to the so--called \emph{one--particle} Hilbert space in the same
context. The use of antiperiodic functions on the torus is related to the
KMS property of equilibrium states and the canonical anticommutation
relations (CAR). The discretization of the torus, leading to $\mathbb{T}_{n}$
for $n\in 2${$\mathbb{N}$}, arises from the use of the Trotter--Kato formula
in the construction of correlation functions of such KMS states as
Berezin--Grassmann integrals.\medskip

\noindent \underline{(ii):} We see the Hilbert space $\mathfrak{h}$ as a
subset of $\ell _{\mathrm{ap}}^{2}(\mathbb{T}_{n};\mathfrak{h})$ by using
the discrete delta function $\delta _{\mathrm{ap}}\in \ell _{\mathrm{ap}%
}^{2}(\mathbb{T}_{n};\mathbb{C})$ defined by%
\begin{equation}
\delta _{\mathrm{ap}}\left( \alpha \right) \doteq \left\{
\begin{array}{lll}
0 & \qquad \text{if}\qquad & \alpha \notin \left\{ 0,\beta \right\} \ . \\
\frac{\beta ^{-1}n}{2}\medskip & \qquad \text{if}\qquad & \alpha =0\ . \\
-\frac{\beta ^{-1}n}{2} & \qquad \text{if}\qquad & \alpha =\beta \ .%
\end{array}%
\right.  \label{define delta}
\end{equation}%
Vectors $\varphi $ of $\mathfrak{h}$ are viewed as antiperiodic functions $%
\hat{\varphi}$ of $\ell _{\mathrm{ap}}^{2}(\mathbb{T}_{n};\mathfrak{h})$ via
the definition%
\begin{equation}
\hat{\varphi}\left( \alpha \right) \doteq \delta _{\mathrm{ap}}\left( \alpha
\right) \varphi \ ,\qquad \alpha \in \mathbb{T}_{n}\ .
\label{define deltabis}
\end{equation}%
Note that this identification is isometric up to a constant, since%
\begin{equation}
\left\langle \hat{\varphi}_{1},\hat{\varphi}_{2}\right\rangle _{\ell _{%
\mathrm{ap}}^{2}(\mathbb{T}_{n};\mathfrak{h})}=\frac{\beta ^{-1}n}{2}%
\left\langle \varphi _{1},\varphi _{2}\right\rangle _{\mathfrak{h}}\ ,\qquad
\varphi _{1},\varphi _{2}\in \mathfrak{h}\ .  \label{dfsdklfjsldkfjs}
\end{equation}

The discrete delta function $\delta _{\mathrm{ap}}$ is useful here because
of the property
\begin{equation}
g\ast \delta _{\mathrm{ap}}=g\ ,\qquad g\in \ell _{\mathrm{ap}}^{2}(\mathbb{T%
}_{n};\mathfrak{h})\ ,  \label{convolution unit}
\end{equation}%
with the convolution being defined by
\begin{equation}
g\ast f\left( \alpha \right) \doteq n^{-1}\beta \sum\limits_{\vartheta \in
\mathbb{T}_{n}}g\left( \alpha -\vartheta \right) f\left( \vartheta \right)
,\ \ g\in \ell _{\mathrm{ap}}^{2}(\mathbb{T}_{n};\mathfrak{h}),\ f\in \ell _{%
\mathrm{ap}}^{2}(\mathbb{T}_{n};\mathbb{C}),\ \alpha \in \mathbb{T}_{n}\ .
\label{convolution}
\end{equation}%
Indeed, $\delta _{\mathrm{ap}}$ is used below to construct the inverse of
some discrete difference operator, see Equation (\ref{eq stupidze}).

\subsection{Discrete--time Covariance\label{sect torus2}}

The discrete--time covariance is an operator defined from (i) a
self--adjoint operator acting on the Hilbert space $\mathfrak{h}$ and (ii)
the discrete derivative operator acting on the space of antiperiodic
functions: \medskip

\noindent \underline{(i):} Any (possibly unbounded) operator $A$ acting on $%
\mathfrak{h}$ with domain $\mathrm{dom}(A)$ is viewed as an operator $\hat{A}
$ with domain
\begin{equation*}
\ell _{\mathrm{ap}}^{2}(\mathbb{T}_{n};\mathrm{dom}(A))\subset \ell _{%
\mathrm{ap}}^{2}(\mathbb{T}_{n};\mathfrak{h})
\end{equation*}%
by the definition
\begin{equation}
\lbrack \hat{A}f]\left( \alpha \right) \doteq A\left( f\left( \alpha \right)
\right) \ ,\qquad f\in \ell _{\mathrm{ap}}^{2}(\mathbb{T}_{n};\mathrm{dom}%
(A)),\ \alpha \in \mathbb{T}_{n}\ .  \label{def A elta1def A elta1}
\end{equation}%
If $A=H=H^{\ast }$ then $\hat{H}$ is also self--adjoint on the Hilbert space
$\ell _{\mathrm{ap}}^{2}(\mathbb{T}_{n};\mathfrak{h})$ of antiperiodic
functions.

The (possibly unbounded) self--adjoint operator $H=H^{\ast }$ acting on the
Hilbert space $\mathfrak{h}$ is viewed\ as the so--called \emph{%
one--particle Hamiltonian} in (fermionic and non--relativistic) quantum
field theory. Indeed, its second quantization refers to the free part of the
full interaction of the fermion system. \medskip

\noindent \underline{(ii):} The discrete derivative operator $\partial \in
\mathcal{B}(\ell _{\mathrm{ap}}^{2}(\mathbb{T}_{n};\mathfrak{h}))$ is the
bounded operator defined by%
\begin{equation}
\partial f\left( \alpha \right) \doteq \beta ^{-1}n\left( f\left( \alpha
+n^{-1}\beta \right) -f\left( \alpha \right) \right) \ ,\qquad f\in \ell _{%
\mathrm{ap}}^{2}(\mathbb{T}_{n};\mathfrak{h}),\ \alpha \in \mathbb{T}_{n}\ .
\label{delta}
\end{equation}%
It is a normal invertible operator. Combining (\ref{def A elta1def A elta1})
and (\ref{delta}) we remark that%
\begin{equation*}
\lbrack \hat{A},\partial ]\doteq \hat{A}\partial -\partial \hat{A}=0
\end{equation*}%
for any operator $A$ acting on $\mathfrak{h}$. Because the discrete
derivative operator $\partial $ acts on a space of antiperiodic functions,%
\begin{equation*}
\inf \mathrm{spec}\left( \left\vert \mathrm{Im}\partial \right\vert \right)
>0\ .
\end{equation*}%
Hence, if $H=H^{\ast }$ is any self--adjoint operator acting on $\mathfrak{h}
$, then $(\partial +\hat{H})$ is a (possibly unbounded) normal operator with
bounded inverse. The \emph{discrete--time covariance} is thus defined to be%
\begin{equation}
C_{H}\doteq -2\left( \partial +\hat{H}\right) ^{-1}\in \mathcal{B}(\ell _{%
\mathrm{ap}}^{2}(\mathbb{T}_{n};\mathfrak{h}))\ .
\label{discrete tim covariance}
\end{equation}

This type of operator appears as the covariance of Gaussian Berezin--Grass%
\-%
mann integrals used in the construction of correlation functions for systems
of interacting fermions, see \cite{Sam}. The discrete--time derivative is
related to the corresponding Trotter--Kato product formula used to define
such integrals, as already mentioned in Section \ref{sect torus1}.

\subsection{Determinant Bounds in Constructive Quantum Field Theory\label%
{Sec Constructive}}

Correlation functions of interacting fermions can be constructed by
perturbation series in the regime of weak couplings. In this context, the
self--adjoint (possibly unbounded) operator $H=H^{\ast }$ acting on $%
\mathfrak{h}$ is the generator of the unperturbed dynamics of the fermion
system.

Now, suppose, for simplicity, that $\mathfrak{h}$ is a \emph{separable}
Hilbert space with ONB $\{\varphi _{\mathfrak{i}}\}_{\mathfrak{i}\in \mathbb{%
I}}$, $\mathbb{I}$ being countable, and set%
\begin{equation}
\mathbf{\omega }_{H,\kappa }\doteq \limsup_{n\rightarrow \infty }\sup_{%
\mathfrak{i}\in \mathbb{I}}\left\{ n^{-1}\beta \sum\limits_{\vartheta \in
\mathbb{T}_{n}}\sum_{q\in \mathbb{I}}\left\vert \left\langle \varphi
_{q},\left( C_{H}\kappa (\hat{H})\hat{\varphi}_{\mathfrak{i}}\right) \left(
\vartheta \right) \right\rangle _{\mathfrak{h}}\right\vert \right\}
\label{explicit decay}
\end{equation}%
for any $\beta \in \mathbb{R}^{+}$, $H=H^{\ast }$ and measurable function $%
\kappa $ from $\mathbb{R}$ to $\mathbb{R}_{0}^{+}$. See (\ref{define
deltabis}), (\ref{def A elta1def A elta1}) and (\ref{discrete tim covariance}%
). We have in mind cutoff functions $\kappa :\mathbb{R\rightarrow }\left[ 0,1%
\right] $.

Another essential quantity in non--relativistic fermionic constructive
quantum field theory is the so--called \emph{determinant bound} of $H$ and $%
\kappa $ defined as follows:

\begin{definition}[Determinant bounds]
\label{determinant bounds}\mbox{ }\newline
The parameter $\gamma _{H,\kappa }\in \mathbb{R}^{+}$ is a determinant bound
of $H=H^{\ast }$ and the measurable function $\kappa :\mathbb{R\rightarrow R}%
_{0}^{+}$ if, for any $\beta \in \mathbb{R}^{+}$, $n\in 2\mathbb{N}$, $%
m,N\in \mathbb{N}$, $\mathfrak{M}\in \mathrm{Mat}\left( m,\mathbb{R}\right) $
with $\mathfrak{M}\geq 0$, and all parameters
\begin{equation*}
\{(\alpha _{q},\mathfrak{i}_{q},j_{q})\}_{q=1}^{2N}\subset \mathbb{T}%
_{n}\cap \lbrack 0,\beta )\times \mathbb{I}\times \{1,\ldots ,m\}\ ,
\end{equation*}%
the following bound holds true:%
\begin{equation}
\left\vert \mathrm{det}\left[ \mathfrak{M}_{j_{k},j_{N+l}}\left\langle
\varphi _{\mathfrak{i}_{N+l}},\left( C_{H}\kappa (\hat{H})\hat{\varphi}_{%
\mathfrak{i}_{k}}\right) \left( \alpha _{k}-\alpha _{N+l}\right)
\right\rangle _{\mathfrak{h}}\right] _{k,l=1}^{N}\right\vert \leq \gamma
_{H,\kappa }^{2N}\prod_{q=1}^{2N}\mathfrak{M}_{j_{q},j_{q}}^{1/2}\ .
\label{Condition super ovni}
\end{equation}
\end{definition}

\noindent For $\mathfrak{M}$ we have in mind positive matrices appearing in
the so--called \emph{Brydges--Kennedy tree expansions} which have the
following structure: For each non--oriented graph $\mathfrak{g}$ with $m$
vertices, all functions $\mathbf{\alpha }\in \left[ 0,1\right] ^{\mathfrak{g}%
}$ and any parameter $s\in \left[ 0,1\right] $, we define the subgraph%
\begin{equation*}
\mathfrak{g}\left( \mathbf{\alpha },s\right) \doteq \mathfrak{g}\backslash
\left\{ \ell \in \mathfrak{g}:\mathbf{\alpha }\left( \ell \right) \geq
s\right\} \subset \mathfrak{g}\ .
\end{equation*}%
In fact, only minimally connected graphs (trees) $\mathfrak{g}$ are relevant
for the Brydges--Kennedy tree expansions. Let $\mathcal{R}_{\mathfrak{g}%
\left( \mathbf{\alpha },s\right) }\subset \left\{ 1,\ldots ,m\right\} ^{2}$
denote the smallest equivalence relation for which one has $(k,l)\in
\mathcal{R}_{\mathfrak{g}\left( \mathbf{\alpha },s\right) }$ for all $k,l\in
\left\{ 1,\ldots ,m\right\} $ such that\ the line $\left\{ k,l\right\} $
belongs to the graph $\mathfrak{g}\left( \mathbf{\alpha },s\right) $. Then,
for any $t\in \left[ 0,1\right] $, $\mathfrak{M}=\mathfrak{M}\left(
\mathfrak{g},\mathbf{\alpha },t\right) $ is the symmetric positive $m\times
m $ real matrix defined by
\begin{equation*}
\left[ \mathfrak{M}\left( \mathfrak{g},\mathbf{\alpha },t\right) \right]
_{k,l}\doteq \int_{0}^{t}\mathbf{1}\left[ (k,l)\in \mathcal{R}_{\mathfrak{g}%
\left( \mathbf{\alpha },s\right) }\right] \mathrm{d}s\ ,\qquad k,l\in
\left\{ 1,\ldots ,m\right\} \ .
\end{equation*}%
See for instance \cite{AR,BK,SW}.

Assume that the matrix entries, with respect to some fix orthonormal basis,
of the interparticle interaction decay sufficiently fast, and let $u\in
\mathbb{R}$ be the coupling constant of the considered interacting fermion
system, i.e., $u$ quantifies the strength of the interparticle interaction.
Then, it can be shown that, if the parameter $\mathbf{\omega }_{H,\mathbf{1}%
_{\mathbb{R}}}\gamma _{H,\mathbf{1}_{\mathbb{R}}}^{2}\left\vert u\right\vert
$ is small enough, the perturbation expansion of \emph{all correlation
functions} in terms of powers of $u$ converges absolutely. More precisely,
all correlation functions are analytic functions of the coupling $u$ at $u=0$
with analyticity radius of order $\mathbf{\omega }_{H,\mathbf{1}_{\mathbb{R}%
}}^{-1}\gamma _{H,\mathbf{1}_{\mathbb{R}}}^{-2}$. See for instance \cite%
{AR,SW}.

The use of the cutoff function $\kappa $ is important in \emph{multiscale
analyses} of correlation functions of interacting fermion systems. Indeed,
even for couplings $\left\vert u\right\vert $ much larger than the
convergence radius $\mathbf{\omega }_{H,\mathbf{1}_{\mathbb{R}}}^{-1}\gamma
_{H,\mathbf{1}_{\mathbb{R}}}^{-2}$ correlations functions can still be
constructed via multiscale schemes related to the Wilson renormalization
group: Take a family $\{\kappa _{L}\}_{L\in \mathbb{N}}$ of measurable
functions from $\mathbb{R}$ to $\left[ 0,1\right] $ such that
\begin{equation*}
\sum_{L=1}^{\infty }\kappa _{L}\left( x\right) =1\ ,\qquad x\in \mathbb{R}\ .
\end{equation*}%
(I.e., the family is a partition of unity.) If $\mathbf{\omega }_{H,\kappa
_{L}}\gamma _{H,\kappa _{L}}^{2}\left\vert u\right\vert $ is small enough
for all $L\in \mathbb{N}$, then, up to technical details, the perturbation
series \emph{at scale} $L$ in terms of powers of $u$ converges absolutely.
In general, the smallness of the parameters $\mathbf{\omega }_{H,\kappa
_{L}}\gamma _{H,\kappa _{L}}^{2}\left\vert u\right\vert $ at all scales is a
much weaker condition than the smallness of $\mathbf{\omega }_{H,\mathbf{1}_{%
\mathbb{R}}}\gamma _{H,\mathbf{1}_{\mathbb{R}}}^{2}\left\vert u\right\vert $%
. See for instance \cite{Walter}.

Note that the form of cutoff function we consider does not depend on the $%
\alpha $ variables, that is, the dependency on the Matsubara frequency of
covariance \emph{does not need} to be regularized, in contrast to other
approaches like for instance \cite{GM,BGPS,GMP}.

Indeed, coming back to the estimate of the form (\ref{Condition super ovni}%
), one easily shows from the \emph{Gram bound} for determinants that
\begin{multline*}
\left\vert \mathrm{det}\left[ \mathfrak{M}_{j_{k},j_{N+l}}\left\langle
\varphi _{\mathfrak{i}_{N+l}},\left( C_{H}\kappa (\hat{H})\hat{\varphi}_{%
\mathfrak{i}_{k}}\right) \left( \alpha _{k}-\alpha _{N+l}\right)
\right\rangle _{\mathfrak{h}}\right] _{k,l=1}^{N}\right\vert \\
\leq \left\Vert C_{H}\right\Vert _{\mathcal{B}(\ell _{\mathrm{ap}}^{2}(%
\mathbb{T}_{n};\mathfrak{h}))}^{N}\prod_{q=1}^{2N}\left\Vert \sqrt{\kappa (%
\hat{H})}\hat{\varphi}_{\mathfrak{i}_{q}}\right\Vert _{\ell _{\mathrm{ap}%
}^{2}(\mathbb{T}_{n};\mathfrak{h})}\mathfrak{M}_{j_{q},j_{q}}^{1/2}\ .
\end{multline*}%
This kind of estimate gives no finite determinant bound of $H$ and $\kappa $
because, in general, the norm of $C_{H}$ diverges, as $n\rightarrow \infty $%
. This problem appears already for \emph{bounded} $H\in \mathcal{B}(%
\mathfrak{h})$ when $0\in \mathrm{spec}(H)$, because in this case%
\begin{equation*}
\Vert C_{H}\Vert _{\mathcal{B}(\ell _{\mathrm{ap}}^{2}(\mathbb{T}_{n};%
\mathfrak{h}))}^{1/2}=\mathcal{O}\left( \sqrt{n}\right) \qquad \text{and}%
\qquad \Vert \hat{\varphi}_{\mathfrak{i}_{q}}\Vert _{\ell _{\mathrm{ap}}^{2}(%
\mathbb{T}_{n};\mathfrak{h})}=\mathcal{O}\left( \sqrt{n}\right) \ ,
\end{equation*}%
as $n\rightarrow \infty $. See (\ref{dfsdklfjsldkfjs}). Nevertheless,
similar to the multiscale analysis presented above, one can tackle this
problem by using the Gram bound as previously for some regularized
covariances $C_{H}\hat{\kappa}_{L}(\hat{H},i\partial )$ at every $L\in
\mathbb{N}$. Here, for any $L\in \mathbb{N}$, $\hat{\kappa}_{L}:\mathbb{R}%
^{2}\rightarrow \left[ 0,1\right] $ is some measurable function of two
variables in such a way that
\begin{equation*}
\sum_{L=1}^{\infty }\hat{\kappa}_{L}\left( x,y\right) =\kappa \left(
x\right) \ ,\qquad x,y\in \mathbb{R}\ .
\end{equation*}%
This decomposition can be chosen such that there are constants $\hat{\gamma}%
_{L}\in \mathbb{R}^{+}$, $L\in \mathbb{N}$, which at least \emph{do not
depend} on $n\in 2{\mathbb{N}}$ and meanwhile satisfy%
\begin{equation*}
\left\vert \mathrm{det}\left[ \mathfrak{M}_{j_{k},j_{N+l}}\left\langle
\varphi _{\mathfrak{i}_{N+l}},\left( C_{H}\hat{\kappa}_{L}(\hat{H},i\partial
)\hat{\varphi}_{\mathfrak{i}_{k}}\right) \left( \alpha _{k}-\alpha
_{N+l}\right) \right\rangle _{\mathfrak{h}}\right] _{k,l=1}^{N}\right\vert
\leq \hat{\gamma}_{L}^{2N}\prod_{q=1}^{2N}\mathfrak{M}_{j_{q},j_{q}}^{1/2}\ .
\end{equation*}%
As already mentioned, such a bound follows from the \emph{usual} \emph{Gram
bound} for determinants. This kind of strategy is used for instance in \cite[%
Section 3]{BGPS}, \cite[Section 3.2]{GM}, (more recently) \cite[Section 5.A.]%
{GMP}, and in many others works. \cite{dSPS} shows that this multiscale
analysis for the so--called the Matsubara UV problem is \emph{not}
necessary, by proving a new bound for determinants that generalizes the
original Gram bound, see \cite[Theorem 1.3]{dSPS}. Note finally that using
multiscale analysis to treat the Matsubara UV problem can, moreover, render
useful properties of the \emph{full} covariance less transparent. Hence,
avoiding this kind of procedure brings various technical benefits.

In the same spirit, we derive direct bounds of the type (\ref{Condition
super ovni}) that do not need the UV regularization of the Matsubara
frequency. One technical advantage of the approach we present here is that
the given covariance does not need to be decomposed as in \cite[Eq. (8)]%
{dSPS} in order to obtain determinant bounds. Moreover, our estimates are
\emph{sharp} (or optimal) and hold true for \emph{all} (possibly unbounded,
the latter not being limited to semibounded) one--particle Hamiltonians.
Observe that \cite{dSPS} gives sharp estimates \emph{up to\ a prefactor 2}
for the class of \emph{bounded} operators it applies, see \cite[Theorem 2.4
and discussions below it]{dSPS}.

In this paper we show the (possibly infinite) general bound
\begin{eqnarray}
\mathfrak{x} &\doteq &\sup
\Big{\{}%
\inf \left\{ \gamma _{H,\mathbf{1}_{\mathbb{R}}}\in \mathbb{R}^{+}:\gamma
_{H,\mathbf{1}_{\mathbb{R}}}\text{ determinant bound of }H\text{ and }%
\mathbf{1}_{\mathbb{R}}\right\}  \notag \\
&:&H=H^{\ast }\ \text{acting on a separable Hilbert space }\mathfrak{h}\text{
with ONB }\{\varphi _{\mathfrak{i}}\}_{\mathfrak{i}\in \mathbb{I}}%
\Big{\}}%
\ ,  \notag \\
&&  \label{universal determinant bound}
\end{eqnarray}%
named here the \emph{universal determinant bound}, is equal to $\mathfrak{x}%
=1$. (Even if the class of all separable Hilbert spaces is not a set, the
supremum is well--defined because of the separation axiom.) In particular,
the convergence of perturbation series at $u=0$ of any non--relativistic
fermionic quantum field theory (possibly in the continuum) is ensured by the
smallness of the positive parameter $\mathbf{\omega }_{H,\mathbf{1}_{\mathbb{%
R}}}$, i.e., if the interaction and the covariance are summable, only. To
our knowledge, such estimates are \emph{unknown} for unbounded self--adjoint
operators $H$, even for semibounded ones. Similar statements can also be
derived while taking into account the (cutoff) function $\kappa $, see
Corollary \ref{theorem main super2}. Note that we consider separable Hilbert
spaces in (\ref{universal determinant bound}) to avoid technical issues.

\begin{bemerkung}[Covariance in the continuum]
\label{remark blabla}\mbox{ }\newline
In the continuous case, we would like to stress that, in contrast to the
lattice case, we do not have in mind covariances of the form
\begin{equation*}
\mathbf{c}((x_{1},\alpha _{1}),(x_{2},\alpha _{2}))=\int_{\mathbb{R}^{d}}%
\frac{\mathrm{e}^{(\alpha _{1}-\alpha _{2})E(p)+i\left\langle
p,(x_{1}-x_{2})\right\rangle _{\mathbb{R}^{d}}}}{1+\mathrm{e}^{\beta E(p)}}%
\mathrm{d}^{d}p\ ,
\end{equation*}%
with $(x_{1},\alpha _{1}),(x_{2},\alpha _{2})\in \mathbb{R}^{d}\times
\lbrack 0,\beta )$ and $\alpha _{1}\geq \alpha _{2}$, i.e., Fourier
transforms of the Fermi--Dirac distribution associated with dispersion
relations $E:\mathbb{R}^{d}\rightarrow \mathbb{R}$. Indeed, such functions
generally diverge for $x_{1}=x_{2}$ when $\alpha _{1}-\alpha _{2}$ tends to $%
\beta $ and, hence, cannot have a finite determinant bound. Formally, such
covariances would correspond to use (Dirac) delta functions in (\ref%
{explicit decay}), instead of the orthonormal vectors $\varphi _{\mathfrak{i}%
}$.
\end{bemerkung}

\begin{bemerkung}[Determinant bounds in the continuum]
\label{remark blabla copy(1)}\mbox{ }\newline
For any fixed $\varphi \in L^{2}(\mathbb{R}^{d})$, its Fourier transform
has, of course, to decay at large frequencies. However, we cannot conclude
from this that determinant bounds derived here are related to the
boundedness of spacial frequencies, because the bounds are uniform with
respect to the choice of the unit vectors $\varphi _{\mathfrak{i}}$.
\end{bemerkung}

\section{Main Results\label{Main Results}}

The proofs are based on two consecutive transformations of the determinant
of the left--hand side of Inequality (\ref{Condition super ovni}):

\begin{itemize}
\item[(a)] We first write this determinant as the limit $\nu \rightarrow
\infty $ of correlation functions associated with quasi--free states $\rho
_{S_{\nu }}$. This is reminiscent of \cite[Theorem 3.7]{dSPS}, which
represents determinants as time--ordered correlation functions of Fock
states (a special case of quasi--free state). In contrast to the present
work, \cite[Theorem 3.7]{dSPS} cannot be applied to the full covariance,
but, rather, for each term of the decomposition \cite[Eq. (8)]{dSPS}.

\item[(b)] For any $\nu \in \mathbb{R}^{+}$, these correlation functions are
represented as scalar products involving modular operators in the GNS
representation of $\rho _{S_{\nu }}$. See Equation (\ref{chiot final1}). As
compared to \cite{dSPS}, the representation of the determinant of (\ref%
{Condition super ovni}) obtained from this second transformation has the
advantage of avoiding the decomposition \cite[Eq. (8)]{dSPS}, which can be
non--trivial to verify for general Hamiltonians and lead to artificial
prefactors in the bounds.
\end{itemize}

\noindent These two transformations allow us to get bounds of the form (\ref%
{Condition super ovni}) by using \cite[(A.2)]{Araki-Masuda}, which can be
viewed as H\"{o}lder inequalities for general non--commutative $L^{p}$%
--spaces.

Sections \ref{Sect Tree--Expansions} and \ref{Bernoulli} explain the main
lines of (a). The details of this first transformation are postponed to
Sections \ref{Quasi free sect} and \ref{representation quasi free}. In
Section \ref{Modular}, we give a few key definitions and results on the
Tomita--Takesaki modular theory used for the transformation (b), which is
described in detail in Section \ref{Sectino chiot}. In particular, we
explain the origin of modular objects appearing in our main theorem, that
is, Theorem \ref{theorem main super}. This section is devoted to the readers
who may not be acquainted with the Tomita--Takesaki modular theory. The main
results of this paper, that is, Theorem \ref{theorem main super} and
Corollaries \ref{theorem main super2}--\ref{theorem main super2 copy(1)},
are found in Section \ref{sectino main encore}, while Section \ref{sect
Schatten} illustrates the central arguments of the proofs in the finite
dimensional case via H\"{o}lder inequalities for Schatten norms.

Recall that $\mathfrak{h}$ is an arbitrary \emph{separable} Hilbert space.
In all the section, we fix $\beta \in \mathbb{R}^{+}$, $n\in 2\mathbb{N}$, $%
m\in \mathbb{N}$, $\mathfrak{M}\in \mathrm{Mat}\left( m,\mathbb{R}\right) $
with $\mathfrak{M}\geq 0$, while $H=H^{\ast }$ is any self--adjoint operator
acting on $\mathfrak{h}$. Note again that $H$ must not be bounded. To avoid
triviality of assertions, we assume $\mathfrak{M}\neq 0$.

\subsection{Quasi--Free States Associated with the Determinants of the
Discrete--time Covariance\label{Sect Tree--Expansions}}

The aim of this section is to represent the determinant of (\ref{Condition
super ovni}) in terms of quasi--free states. To this end, we first define
CAR $C^{\ast }$--algebras $\mathrm{CAR}(\mathfrak{h}\otimes \mathbb{M})$
constructed from a fixed $\mathfrak{h}$ and some finite--dimensional Hilbert
spaces $\mathbb{M}$, having in mind the positive matrices $\mathfrak{M}$
appearing in the Brydges--Kennedy tree expansions: \medskip

\noindent \underline{(i):} The (generic) non--vanishing positive matrix $%
\mathfrak{M}$ gives rise to a positive sesquilinear form defined on $\mathbb{%
C}^{m}$ by
\begin{equation}
\left\langle \left( x_{1},\ldots ,x_{m}\right) ,\left( y_{1},\ldots
,y_{m}\right) \right\rangle _{\mathbb{C}^{m}}^{\mathfrak{M}}\doteq
\sum_{p,q=1}^{m}\overline{x_{p}}\ y_{q}\mathfrak{M}_{p,q}\ .
\label{scalar product}
\end{equation}%
In general, this sesquilinear form is degenerated. The vector space $\mathbb{%
M}$ is then defined to be the quotient
\begin{equation*}
\mathbb{M}\doteq \mathbb{C}^{m}/\{x\in \mathbb{C}^{m}:\left\langle
x,x\right\rangle _{\mathbb{C}^{m}}^{\mathfrak{M}}=0\}\ .
\end{equation*}%
Then, as usual, we introduce a scalar product on $\mathbb{M}$ as
\begin{equation*}
\left\langle \left[ x\right] ,\left[ y\right] \right\rangle _{\mathbb{M}%
}\doteq \left\langle x,y\right\rangle _{\mathbb{C}^{m}}^{\mathfrak{M}}\
,\qquad x,y\in \mathbb{C}^{m}\ ,
\end{equation*}%
and $\mathbb{M}$ denotes the Hilbert space $(\mathbb{M},\langle \cdot ,\cdot
\rangle _{\mathbb{M}})$. Using the notation $\mathfrak{e}_{k}\doteq \left[
e_{k}\right] \in \mathbb{M}$, where $\left\{ e_{k}\right\} _{k=1}^{m}$ is
the canonical basis of $\mathbb{C}^{m}$, note that
\begin{equation}
\mathfrak{M}_{k,l}=\left\langle \mathfrak{e}_{k},\mathfrak{e}%
_{l}\right\rangle _{\mathbb{M}}\ ,\qquad k,l\in \left\{ 1,\ldots ,m\right\}
\ .  \label{ej fract}
\end{equation}%
\medskip

\noindent \underline{(ii):} The (extended) CAR $C^{\ast }$--algebra
associated with $\mathfrak{M}$ is the unital $C^{\ast }$--algebra $\mathrm{%
CAR}(\mathfrak{h}\otimes \mathbb{M})$ generated by the unit $\mathbf{1}$ and
the family $\{a(\Psi )\}_{\Psi \in \mathfrak{h}\otimes \mathbb{M}}$ of
elements satisfying the canonical anticommutation relations (CAR), see (\ref%
{CAR AA})--(\ref{CAR AA*}) with $\mathcal{H}=\mathfrak{h}\otimes \mathbb{M}$%
. Notice that such a family always exists and two families satisfying these
CAR are related to each other by a unique $\ast $--automorphism on the $%
C^{\ast }$--algebra $\mathrm{CAR}(\mathfrak{h}\otimes \mathbb{M})$. See,
e.g., \cite[Theorem 5.2.5]{BratteliRobinson}.

The element $a(\Psi )\in \mathrm{CAR}(\mathfrak{h}\otimes \mathbb{M})$ is,
in fermionic quantum field theory, the annihilation operator associated with
$\Psi \in \mathfrak{h}\otimes \mathbb{M}$ whereas its adjoint
\begin{equation*}
a^{+}(\Psi )\doteq a(\Psi )^{\ast }\ ,\qquad \Psi \in \mathfrak{h}\otimes
\mathbb{M}\ ,
\end{equation*}%
is the corresponding creation operator.

Considering that $\mathfrak{h}$ represents the one--particle Hilbert space, $%
\mathrm{CAR}(\mathfrak{h})$ is the $C^{\ast }$--algebra that allows to
represent the corresponding many--fermion system within the algebraic
formulation of quantum mechanics. The extension of this $C^{\ast }$--algebra
to $\mathrm{CAR}(\mathfrak{h}\otimes \mathbb{M})$ is pivotal to control the
determinant of (\ref{Condition super ovni}). Such determinants are naturally
expressed through limits of \emph{quasi--free states} on the $C^{\ast }$%
--algebra $\mathrm{CAR}(\mathfrak{h}\otimes \mathbb{M})$: Quasi--free states
are positive linear functionals $\rho \in \mathrm{CAR}(\mathfrak{h}\otimes
\mathbb{M})^{\ast }$ such that $\rho (\mathbf{1})=1$ and, for all $%
N_{1},N_{2}\in \mathbb{N}$ and $\Psi _{1},\ldots ,\Psi _{N_{1}+N_{2}}\in
\mathfrak{h}\otimes \mathbb{M}$,%
\begin{equation}
\rho \left( a^{+}(\Psi _{1})\cdots a^{+}(\Psi _{N_{1}})a(\Psi
_{N_{1}+N_{2}})\cdots a(\Psi _{N_{1}+1})\right) =0  \label{ass O0-00}
\end{equation}%
if $N_{1}\neq N_{2}$, while in the case $N_{1}=N_{2}\equiv N$,
\begin{equation}
\rho \left( a^{+}(\Psi _{1})\cdots a^{+}(\Psi _{N})a(\Psi _{2N})\cdots
a(\Psi _{N+1})\right) =\mathrm{det}\left[ \rho \left( a^{+}(\Psi _{k})a(\Psi
_{N+l})\right) \right] _{k,l=1}^{N}\,.  \label{ass O0-00bis}
\end{equation}

\begin{bemerkung}[Other definitions of quasi--free states in the literature]

\mbox{ }\newline
Some authors relax Condition (\ref{ass O0-00}) in the definition of
quasi--free states. Within this more general framework (known as the
self--dual formalism) quasi--free states fulfilling (\ref{ass O0-00}) are
then referred as gauge invariant quasi--free states of the corresponding CAR
$C^{\ast }$--algebras. For instance, see \cite[Definition 3.1]{Araki}. Note
indeed that \cite[Definition 3.1, Condition (3.1)]{Araki} only imposes on
the quasi--free state to be even, but not necessarily gauge invariant.
\end{bemerkung}

The operator $S^{(\rho )}\in \mathcal{B}(\mathfrak{h}\otimes \mathbb{M})$
defined from%
\begin{equation}
\left\langle \Psi _{2},S^{(\rho )}\Psi _{1}\right\rangle _{\mathfrak{h}%
\otimes \mathbb{M}}=\rho \left( a^{+}(\Psi _{1})a(\Psi _{2})\right) \
,\qquad \Psi _{1},\Psi _{2}\in \mathfrak{h}\otimes \mathbb{M}\ ,
\label{ass O0-00bisbis}
\end{equation}%
is named the \emph{symbol} (or one--particle density matrix) of the
quasi--free state $\rho $. By the positivity and normalization of states, it
follows that symbols are positive (self--adjoint) operators with spectrum
lying on the unit interval $[0,1]$. Conversely, any such positive operator $%
S\leq \mathbf{1}_{\mathfrak{h}\otimes \mathbb{M}}$ on $\mathfrak{h}\otimes
\mathbb{M}$ uniquely defines a quasi--free state $\rho _{S}$ on $\mathrm{CAR}%
(\mathfrak{h}\otimes \mathbb{M})$ such that
\begin{equation}
\rho _{S}\left( a^{+}(\Psi _{1})a(\Psi _{2})\right) =\left\langle \Psi
_{2},S\Psi _{1}\right\rangle _{\mathfrak{h}\otimes \mathbb{M}}\ ,\qquad \Psi
_{1},\Psi _{2}\in \mathfrak{h}\otimes \mathbb{M}\ .
\label{quasi free symbol}
\end{equation}

The symbols allowing us to represent the determinant of (\ref{Condition
super ovni}) in terms of quasi--free states are defined as follows: For all $%
\nu \in \mathbb{R}^{+}$, define the function $\ $%
\begin{equation}
\digamma _{\nu }\left( \lambda \right) \doteq \left\{
\begin{array}{lll}
-\beta ^{-1}n\ln \left\vert 1-n^{-1}\beta \lambda \right\vert & \qquad \text{%
if}\qquad & \lambda \in \mathbb{R}\backslash \{\beta ^{-1}n\}\ , \\
\nu & \qquad \text{if}\qquad & \lambda =\beta ^{-1}n\ ,%
\end{array}%
\right.  \label{function F}
\end{equation}%
and let
\begin{equation}
H_{\nu }\doteq \digamma _{\nu }\left( H\right) \ ,\qquad \nu \in \mathbb{R}%
^{+}\ .  \label{Hn0Hn0}
\end{equation}%
The relevant quasi--free states on the $C^{\ast }$--algebra $\mathrm{CAR}(%
\mathfrak{h}\otimes \mathbb{M})$ are those with symbol
\begin{equation}
S_{\nu }\doteq \frac{1}{1+\mathrm{e}^{\beta H_{\nu }\otimes \mathbf{1}_{%
\mathbb{M}}}}=\frac{1}{1+\mathrm{e}^{\beta H_{\nu }}}\otimes \mathbf{1}_{%
\mathbb{M}}\in \mathcal{B}\left( \mathfrak{h}\otimes \mathbb{M}\right) \
,\quad \nu \in \mathbb{R}^{+},  \label{symbol}
\end{equation}%
observing that $0<S_{\nu }\leq \mathbf{1}_{\mathfrak{h}\otimes \mathbb{M}}$.
The precise relationship between the quasi--free states $\rho _{S_{\nu }}$, $%
\nu \in \mathbb{R}^{+}$, and the covariance appearing in the determinant of (%
\ref{Condition super ovni}) is described below.

\subsection{Discrete--time Covariance and Bernoulli--Euler Approximations
\label{Bernoulli}}

At fixed $\lambda \in \mathbb{R}$ and large $n\gg 1$, note from (\ref%
{function F}) that
\begin{equation}
\mathrm{e}^{\mp \beta \digamma _{\nu }\left( \lambda \right) }=\left(
1-n^{-1}\beta \lambda \right) ^{\pm n}=\mathrm{e}^{\mp \beta \lambda
}+o\left( 1\right)  \label{equation idiote}
\end{equation}%
is the well--known Bernoulli--Euler approximation of the exponential
function $\mathrm{e}^{\mp \beta \lambda }$. In particular, $H_{\nu }$, as
defined by (\ref{Hn0Hn0}), can be viewed as an approximation of the
self--adjoint operator $H$. The relevance of the function $\digamma _{\nu }$
results from the following observations: \medskip

\noindent \underline{(i):} By the spectral theorem, there is a ($\sigma $%
--finite) measure space $(\Omega _{H},\mathfrak{A}_{H},\mu _{H})$, a unitary
map $U_{H}$ from $\mathfrak{h}$ to $L^{2}(\Omega _{H};\mathbb{C})$ and a $%
\mathfrak{A}_{H}$--measurable function $\lambda _{H}:\Omega _{H}\rightarrow
\mathbb{R}$ such that%
\begin{equation}
U_{H}HU_{H}^{\ast }=\mathrm{m}_{\lambda _{H}}\ ,  \label{spectral theorem}
\end{equation}%
where $\mathrm{m}_{\lambda _{H}}$ is the multiplication operator on $%
L^{2}(\Omega _{H};\mathbb{C})$ with the function $\lambda _{H}$. Using the
unitary $U_{H}$ we can identify $\ell _{\mathrm{ap}}^{2}(\mathbb{T}_{n};%
\mathfrak{h})$ with $\ell _{\mathrm{ap}}^{2}(\mathbb{T}_{n};L^{2}(\Omega
_{H};\mathbb{C}))$, i.e.,
\begin{equation*}
\hat{U}_{H}\ell _{\mathrm{ap}}^{2}(\mathbb{T}_{n};\mathfrak{h})=\ell _{%
\mathrm{ap}}^{2}(\mathbb{T}_{n};L^{2}(\Omega _{H};\mathbb{C}))\ .
\end{equation*}%
Recall that $\hat{A}$ is the extension to $\ell _{\mathrm{ap}}^{2}(\mathbb{T}%
_{n};\mathfrak{h})$ of any operator $A$ acting on $\mathfrak{h}$, as defined
by (\ref{def A elta1def A elta1}). The latter, in turn, is canonically
identified with%
\begin{equation}
\int_{\Omega _{H}}^{\oplus }\ell _{\mathrm{ap}}^{2}(\mathbb{T}_{n};\mathbb{C}%
)\mu _{H}\left( \mathrm{d}\mathfrak{a}\right) \equiv L^{2}(\Omega _{H};\ell
_{\mathrm{ap}}^{2}(\mathbb{T}_{n};\mathbb{C}))\ .  \label{space cool}
\end{equation}%
In other words, by using $U_{H}$, we identify $\ell _{\mathrm{ap}}^{2}(%
\mathbb{T}_{n};\mathfrak{h})$ with (\ref{space cool}). Note that the above
direct integral is well--defined because $\ell _{\mathrm{ap}}^{2}(\mathbb{T}%
_{n};\mathbb{C})$ is finite dimensional and $(\Omega _{H},\mathfrak{A}%
_{H},\mu _{H})$ is a $\sigma $--finite measure space, since $\mathfrak{h}$
is assumed to be separable. \medskip

\noindent \underline{(ii):} With this convention,%
\begin{equation*}
\hat{U}_{H}\hat{H}\hat{U}_{H}^{\ast }=\int_{\Omega _{H}}^{\oplus }\lambda
_{H}\left( \mathfrak{a}\right) \mathbf{1}_{\ell _{\mathrm{ap}}^{2}(\mathbb{T}%
_{n};\mathbb{C})}\ \mu _{H}\left( \mathrm{d}\mathfrak{a}\right) \ .
\end{equation*}%
The discrete derivative $\partial $ defined by (\ref{delta}) is meanwhile
written in the new Hilbert space as%
\begin{equation*}
\hat{U}_{H}\ \partial \ \hat{U}_{H}^{\ast }=\int_{\Omega _{H}}^{\oplus }%
\mathfrak{d}\ \mu _{H}\left( \mathrm{d}\mathfrak{a}\right) \ ,
\end{equation*}%
where $\mathfrak{d}\in \mathcal{B}(\ell _{\mathrm{ap}}^{2}(\mathbb{T}_{n};%
\mathbb{C}))$ is defined by
\begin{equation*}
\mathfrak{d}f\left( \alpha \right) \doteq \beta ^{-1}n\left( f\left( \alpha
+n^{-1}\beta \right) -f\left( \alpha \right) \right) ,\qquad f\in \ell _{%
\mathrm{ap}}^{2}(\mathbb{T}_{n};\mathbb{C}),\ \alpha \in \mathbb{T}_{n}\ .
\end{equation*}%
In particular, the discrete--time covariance $C_{H}$, defined by (\ref%
{discrete tim covariance}), can be represented as%
\begin{equation}
\hat{U}_{H}C_{H}\hat{U}_{H}^{\ast }=-2\int_{\Omega _{H}}^{\oplus }R\left(
\mathfrak{d},\lambda _{H}\left( \mathfrak{a}\right) \right) \mu _{H}\left(
\mathrm{d}\mathfrak{a}\right) \ ,  \label{Integral direct C}
\end{equation}%
where $R\left( \mathfrak{d},\lambda \right) \in \mathcal{B}(\ell _{\mathrm{ap%
}}^{2}(\mathbb{T}_{n};\mathbb{C}))$ is the resolvent
\begin{equation*}
R\left( \mathfrak{d},\lambda \right) \doteq \left( \mathfrak{d}+\lambda \
\mathbf{1}_{\ell _{\mathrm{ap}}^{2}(\mathbb{T}_{n};\mathbb{C})}\right)
^{-1}\ ,\qquad \lambda \in {\mathbb{R}}\ .
\end{equation*}%
\medskip

\noindent \underline{(iii):} It is convenient to represent the last
resolvent as a convolution (\ref{convolution}) with an antiperiodic
function. To this end, we solve the following equation
\begin{equation}
-2R\left( \mathfrak{d},\lambda \right) f=g_{\lambda }\ast f\ ,\qquad f\in
\ell _{\mathrm{ap}}^{2}(\mathbb{T}_{n};\mathbb{C})\ ,
\label{jean sans bras0}
\end{equation}%
in $g_{\lambda }\in \ell _{\mathrm{ap}}^{2}(\mathbb{T}_{n};\mathbb{C})$ for
any fixed $\lambda \in {\mathbb{R}}$. (Compare with (\ref{Integral direct C}%
).)\medskip

\noindent \underline{(iii.a):} For $\lambda \neq \beta ^{-1}n$ and $\nu \in
\mathbb{R}^{+}$, the antiperiodic function $g_{\lambda }\in \ell _{\mathrm{ap%
}}^{2}(\mathbb{T}_{n};\mathbb{C})$ defined by%
\begin{equation}
g_{\lambda }\left( \alpha \right) \doteq \frac{\left( 1-n^{-1}\beta \lambda
\right) ^{\beta ^{-1}n(\alpha -n^{-1}\beta )}}{1+\mathrm{e}^{\beta \digamma
_{\nu }\left( \lambda \right) }}\ ,\qquad \alpha \in \mathbb{T}_{n}\cap
\left( -\beta ,0\right] \ ,  \label{g explic1}
\end{equation}%
is the unique solution on $\ell _{\mathrm{ap}}^{2}(\mathbb{T}_{n};\mathbb{C}%
) $ of the difference equation%
\begin{equation}
\mathfrak{d}f\left( \alpha \right) +\lambda f\left( \alpha \right) =-2\delta
_{\mathrm{ap}}\left( \alpha \right) \ ,\qquad \alpha \in \mathbb{T}_{n}\ ,
\label{eq stupidze}
\end{equation}%
with the discrete delta function $\delta _{\mathrm{ap}}\in \ell _{\mathrm{ap}%
}^{2}(\mathbb{T}_{n};\mathbb{C})$ being defined by (\ref{define delta}). In
particular, $g_{\lambda }\in \ell _{\mathrm{ap}}^{2}(\mathbb{T}_{n};\mathbb{C%
})$ solves (\ref{jean sans bras0}) for $\lambda \neq \beta ^{-1}n$.

Note that we take $n\in 2\mathbb{N}$ to ensure that
\begin{equation*}
\left( 1-n^{-1}\beta \lambda \right) ^{n}=\left\vert 1-n^{-1}\beta \lambda
\right\vert ^{n}=\mathrm{e}^{-\beta \digamma _{\nu }\left( \lambda \right) }
\end{equation*}%
and observe meanwhile that $\alpha \beta ^{-1}n\in \mathbb{Z}$ if $\alpha
\in \mathbb{T}_{n}$. Therefore, for any $\lambda \neq \beta ^{-1}n$ and $\nu
\in \mathbb{R}^{+}$,
\begin{equation}
g_{\lambda }\left( \alpha \right) =\left( \mathrm{sgn}\left( 1-n^{-1}\beta
\lambda \right) \right) ^{\beta ^{-1}n(n^{-1}\beta -\alpha )}\frac{\mathrm{e}%
^{-(\alpha -n^{-1}\beta )\digamma _{\nu }\left( \lambda \right) }}{1+\mathrm{%
e}^{\beta \digamma _{\nu }\left( \lambda \right) }}\ ,\quad \alpha \in
\mathbb{T}_{n}\cap \left( -\beta ,0\right] \ .  \label{g explic1bis}
\end{equation}%
Recall that $\mathrm{sgn}$ is the sign function defined here as follows: $%
\mathrm{sgn}\left( x\right) \doteq 1$ for $x\in \mathbb{R}_{0}^{+}$ and $%
\mathrm{sgn}\left( x\right) \doteq -1$ otherwise.\medskip

\noindent \underline{(iii.b):} For $\lambda =\beta ^{-1}n$, the (unique)
solution on $\ell _{\mathrm{ap}}^{2}(\mathbb{T}_{n};\mathbb{C})$ of the
difference equation (\ref{eq stupidze}) is equal to
\begin{equation*}
g_{\beta ^{-1}n}\left( \alpha \right) \doteq \left\{
\begin{array}{lll}
0 & \qquad \text{if}\qquad & \alpha \in \mathbb{T}_{n}\backslash \left\{
n^{-1}\beta ,-\beta +n^{-1}\beta \right\} \ . \\
-1 & \qquad \text{if}\qquad & \alpha =n^{-1}\beta \ . \\
1 & \qquad \text{if}\qquad & \alpha =-\beta +n^{-1}\beta \ .%
\end{array}%
\right. \ .
\end{equation*}%
We can write this function as the following limit:
\begin{equation}
g_{\beta ^{-1}n}\left( \alpha \right) =\left\{
\begin{array}{lll}
0\medskip & \text{if} & \alpha \in \mathbb{T}_{n}\backslash \left\{
n^{-1}\beta ,n^{-1}\beta -\beta \right\} \ . \\
-\lim_{\nu \rightarrow \infty }\frac{\mathrm{e}^{(\beta -(\alpha
-n^{-1}\beta ))\digamma _{\nu }(\beta ^{-1}n)}}{1+\mathrm{e}^{\beta \digamma
_{\nu }(\beta ^{-1}n)}}\medskip & \text{if} & \alpha =n^{-1}\beta \ . \\
\lim_{\nu \rightarrow \infty }\frac{\mathrm{e}^{-(\alpha -n^{-1}\beta
)\digamma _{\nu }(\beta ^{-1}n)}}{1+\mathrm{e}^{\beta \digamma _{\nu }(\beta
^{-1}n)}} & \text{if} & \alpha =n^{-1}\beta -\beta \ .%
\end{array}%
\right.  \label{g explic2}
\end{equation}%
In particular, $g_{\lambda }\in \ell _{\mathrm{ap}}^{2}(\mathbb{T}_{n};%
\mathbb{C})$ solves (\ref{jean sans bras0}) for $\lambda =\beta ^{-1}n$.
Compare also (\ref{g explic2}) with (\ref{g explic1bis}). \medskip

\noindent \underline{(iv):} The relationship between the function $%
g_{\lambda }\in \ell _{\mathrm{ap}}^{2}(\mathbb{T}_{n};\mathbb{C})$ and the
symbols $S_{\nu }$ (\ref{symbol}) defining the quasi--free states $\rho
_{S_{\nu }}$, $\nu \in \mathbb{R}^{+}$, can be heuristically understood by
considering the limit case $n=\infty $:\medskip

\noindent \underline{(iv.a):} The function $g_{\lambda }\in \ell _{\mathrm{ap%
}}^{2}(\mathbb{T}_{n};\mathbb{C})$ plays the role, in the discrete case ($%
n<\infty $), of the antiperiodic function $g_{\lambda }^{(\infty )}:\mathbb{%
R\rightarrow R}$ defined by%
\begin{equation}
g_{\lambda }^{(\infty )}\left( \alpha \right) \doteq \frac{\mathrm{e}%
^{-\alpha \lambda }}{1+\mathrm{e}^{\beta \lambda }},\qquad \alpha \in \left(
-\beta ,0\right] \ ,  \label{g exp}
\end{equation}%
which solves the differential equation%
\begin{equation*}
y^{\prime }+\lambda y=\sum_{l=-\infty }^{\infty }\left( -1\right)
^{l+1}\delta _{\beta l}\ .
\end{equation*}%
Here, $\delta _{x}$ is the delta distribution at $x\in \mathbb{R}$. Compare
the last equation with (\ref{eq stupidze}). Up to the observation (\ref%
{equation idiote}) and the special case $\lambda =\beta ^{-1}n$, the
qualitative difference between (\ref{g exp}) and (\ref{g explic1bis})
concerns the replacement of $\alpha $ in (\ref{g exp}) by $\alpha
-n^{-1}\beta $ in (\ref{g explic1bis}) and the prefactor
\begin{equation*}
\left( \mathrm{sgn}\left( 1-n^{-1}\beta \lambda \right) \right) ^{\beta
^{-1}n\left( \alpha -n^{-1}\beta \right) }\ .
\end{equation*}%
\medskip

\noindent \underline{(iv.b):} Using the symbol
\begin{equation*}
S_{H}\doteq \frac{1}{1+\mathrm{e}^{\beta H\otimes \mathbf{1}_{\mathbb{M}}}}=%
\frac{1}{1+\mathrm{e}^{\beta H}}\otimes \mathbf{1}_{\mathbb{M}}\in \mathcal{B%
}\left( \mathfrak{h}\otimes \mathbb{M}\right) \ ,
\end{equation*}%
for any $\alpha _{1},\alpha _{2}\in \mathbb{T}_{\infty }\doteq \left( -\beta
,\beta \right] $ (seen as a torus) with $\alpha _{1}\leq \alpha _{2}$, all
entire analytic vectors $\varphi _{1},\varphi _{2}$ of $H$ and every $%
j_{1},j_{2}\in \{1,\ldots ,m\}$,
\begin{equation}
\rho _{S_{H}}%
\Big(%
a^{+}\left( (\mathrm{e}^{-\alpha _{1}H}\varphi _{1})\otimes \mathfrak{e}%
_{j_{1}}\right) a\left( (\mathrm{e}^{\alpha _{2}H}\varphi _{2})\otimes
\mathfrak{e}_{j_{2}}\right)
\Big)%
=\mathfrak{M}_{j_{1},j_{2}}\left\langle \varphi _{2},\frac{\mathrm{e}%
^{(\alpha _{2}-\alpha _{1})H}}{1+\mathrm{e}^{\beta H}}\varphi
_{1}\right\rangle _{\mathfrak{h}}  \notag
\end{equation}%
with $\mathfrak{e}_{j}\doteq \left[ e_{j}\right] \in \mathbb{M}$ being the
vectors of $\mathbb{M}$ satisfying (\ref{ej fract}). The symbol $S_{H}$ is
directly related to the the antiperiodic function $g_{\lambda }^{(\infty )}$
since
\begin{equation*}
U_{H}S_{H}U_{H}^{\ast }=\int_{\Omega _{H}}^{\oplus }g_{\lambda _{H}\left(
\mathfrak{a}\right) }^{(\infty )}\left( 0\right) \mathbf{1}_{\mathbb{C}}\
\mu _{H}\left( \mathrm{d}\mathfrak{a}\right) \ .
\end{equation*}%
Similar identities hold true in the discrete case for which $S_{H}$ and $%
g_{\lambda }^{(\infty )}$ are replaced with $S_{\nu }$ (\ref{symbol}) and $%
g_{\lambda }$ (\ref{g explic1bis})--(\ref{g explic2}). In particular, the
determinant of (\ref{Condition super ovni}) can be represented in terms of a
limit $\nu \rightarrow \infty $ (cf. (\ref{g explic2})) of quasi-free states
$\rho _{S_{\nu }}$ with symbol $S_{\nu }$ (\ref{symbol}). See Lemma \ref%
{lemma exp copy(1)} and Corollary \ref{Corollary chiot1}.

\subsection{Modular Objects Associated with Discrete--time Covariance\label%
{Modular}}

Our estimates are based on non--commutative H\"{o}lder inequalities \cite[%
(A.2)]{Araki-Masuda} (see also (\ref{property araki chiot 1})), which
requires the celebrated Tomita--Takesaki (modular) theory. Modular objects
associated with discrete--time covariance are constructed, for any fixed $%
\nu \in \mathbb{R}^{+}$, from the quasi--free state $\rho _{S_{\nu }}$ with
symbol $S_{\nu }$ (\ref{symbol}) as follows:\medskip

\noindent \underline{(i):} Let $(\mathfrak{H}_{\nu },\varkappa _{\nu },\eta
_{\nu })$ be a cyclic\ representation of $\rho _{S_{\nu }}$. The weak
closure of the $C^{\ast }$--algebra $\mathrm{CAR}(\mathfrak{h}\otimes
\mathbb{M})$ is the von Neumann algebra
\begin{equation}
\mathcal{X}_{\nu }\doteq \varkappa _{\nu }\left( \mathrm{CAR}\left(
\mathfrak{h}\otimes \mathbb{M}\right) \right) ^{\prime \prime }\subset
\mathcal{B}\left( \mathfrak{H}_{\nu }\right) \ .  \label{sdfjkl}
\end{equation}
As is usual, $\mathfrak{M}^{\prime \prime }$ denotes the bicommutant of any
subset $\mathfrak{M}$ of the space of bounded operators acting on a Hilbert
space. \medskip

\noindent \underline{(ii):} The vector $\eta _{\nu }$ is, by assumption, a
cyclic vector for $\mathcal{X}_{\nu }$, i.e., $\mathcal{H}$ is the closure
of (the linear span of) the set
\begin{equation*}
\mathcal{X}_{\nu }\eta _{\nu }\doteq \left\{ A\eta _{\nu }:A\in \mathcal{X}%
_{\nu }\right\} \ .
\end{equation*}%
Because the vector $\eta _{\nu }$ represents a KMS state (see Section \ref%
{Sectino chiot}), it is also separating for $\mathcal{X}_{\nu }$, i.e., for
all $A\in \mathcal{X}_{\nu }$, $A\eta _{\nu }=0$ iff $A=0$. \medskip

\noindent \underline{(iii):} We define two anti--linear operators $\mathcal{S%
}_{0}$ and $\mathcal{F}_{0}$ respectively by
\begin{equation*}
\mathcal{S}_{0}A\eta _{\nu }=A^{\ast }\eta _{\nu }\qquad \text{and}\qquad
\mathcal{F}_{0}B\eta _{\nu }=B^{\ast }\eta _{\nu }
\end{equation*}%
for any $A\in \mathcal{X}_{\nu }$ and $B\in \mathcal{X}_{\nu }^{\prime }$.
Since a cyclic and separating vector for $\mathcal{X}_{\nu }$ is also cyclic
and separating for its commutant $\mathcal{X}_{\nu }^{\prime }$, both
operators are well--defined on the dense domains $\mathrm{Dom}(\mathcal{S}%
_{0})=\mathcal{X}_{\nu }\eta _{\nu }$ and $\mathrm{Dom}(\mathcal{F}_{0})=%
\mathcal{X}_{\nu }^{\prime }\eta _{\nu }$. By \cite[Proposition 2.5.9]%
{BratteliRobinsonI}, $\mathcal{S}_{0}$ and $\mathcal{F}_{0}$ are closable
and their closure are denoted by $\mathcal{S}$ and $\mathcal{F}$,
respectively. In fact, $\mathcal{F}=\mathcal{S}^{\ast }$ and $\mathcal{S}=%
\mathcal{F}^{\ast }$. \medskip

\noindent \underline{(iv):} The modular operator $\mathbf{\Delta }_{\nu }$
and conjugation $J_{\nu }$ associated with the pair $(\mathcal{X}_{\nu
},\eta _{\nu })$ are respectively the unique, positive, self--adjoint
operator and the unique anti--unitary operator occurring in the polar
decomposition of $\mathcal{S}=J_{\nu }\mathbf{\Delta }_{\nu }^{1/2}$. The
main result of the modular Tomita--Takesaki theory is the Tomita--Takesaki
theorem \cite[Theorem 2.5.14]{BratteliRobinsonI}, which states in the
current context that%
\begin{equation*}
J_{\nu }\mathcal{X}_{\nu }J_{\nu }=\mathcal{X}_{\nu }^{\prime }\qquad \text{%
and}\qquad \mathbf{\Delta }_{\nu }^{it}\mathcal{X}_{\nu }\mathbf{\Delta }%
_{\nu }^{-it}=\mathcal{X}_{\nu }
\end{equation*}%
for all $t\in \mathbb{R}$. The second assertion is related with the
so--called modular automorphism group, as defined by (\ref{definition thobis}%
) in its $\beta $--rescaled version. \medskip

For more details on the theory of von Neumann algebras and modular objects,
see for instance \cite{BratteliRobinsonI}. To make its key points more
transparent, this theory is illustrated in the finite dimensional case in
Section \ref{sect Schatten}. In the same spirit, the non--commutative H\"{o}%
lder inequalities \cite[(A.2)]{Araki-Masuda}, corresponding here to (\ref%
{property araki chiot 1}), are derived in the finite dimensional case from H%
\"{o}lder inequalities for Schatten norms. See (\ref{Holder jean sans bras}%
)--(\ref{Holder finite dimHolder finite dim}).

\subsection{Determinant Bounds from Non--commutative H\"{o}lder Inequalities
\label{sectino main encore}}

To prove our estimates, we rewrite the determinant of (\ref{Condition super
ovni}) by using cyclic representations of quasi--free states on the $C^{\ast
}$--algebra $\mathrm{CAR}(\mathfrak{h}\otimes \mathbb{M})$, as explained in
Section \ref{Sect Tree--Expansions}. This allows us to use the bound \cite[%
(A.2)]{Araki-Masuda}, which can be viewed as H\"{o}lder inequalities for
general non--commutative $L^{p}$--spaces. This yields the following
assertions on determinants of fermionic covariances:

\begin{satz}[Representation of determinants of fermionic covariances]
\label{theorem main super}\mbox{
}\newline
Let $\mathfrak{h}$ be any separable Hilbert space. Take $\beta \in \mathbb{R}%
^{+}$, $m\in \mathbb{N}$, $n\in 2{\mathbb{N}}$, any self--adjoint operator $%
H=H^{\ast }$ acting on $\mathfrak{h}$, and a non--vanishing $\mathfrak{M}\in
\mathrm{Mat}\left( m,\mathbb{R}\right) $ with $\mathfrak{M}\geq 0$. Then
there are von Neumann algebras $\mathcal{X}_{\nu }\subset \mathcal{B}(%
\mathfrak{H}_{\nu })$, cyclic and separating unit vectors $\eta _{\nu }\in
\mathfrak{H}_{\nu }$ (for $\mathcal{X}_{\nu }$) and $C^{\ast }$%
--homomorphisms $\varkappa _{\nu }$ (from $\mathrm{CAR}(\mathfrak{h}\otimes
\mathbb{M})$ to $\mathcal{X}_{\nu }$), where $\nu \in \mathbb{R}^{+}$, such
that for each bounded measurable positive function $\kappa $ from $\mathbb{R}
$ to $\mathbb{R}_{0}^{+}$, all parameters
\begin{equation*}
\{(\alpha _{q},\varphi _{q},j_{q})\}_{q=1}^{2N}\subset \mathbb{T}_{n}\cap
\lbrack 0,\beta )\times \mathfrak{h}\times \{1,\ldots ,m\}\ ,
\end{equation*}%
and for any permutation $\pi $ of $2N\in \mathbb{N}$ elements with sign $%
\left( -1\right) ^{\pi }$ so that\footnote{%
The conditions on $\pi $ impose that it is a permutation of $2N$ elements
which orders the numbers $\alpha _{q}$, $q\in \{1,\ldots ,2N\}$, in the
following way: $\pi (k)<\pi (l)$ whenever $\alpha _{k}<\alpha _{l}$ for $%
k,l\in \{1,\ldots ,2N\}$ while $\pi (k)<\pi (N+l)$ whenever $\alpha
_{k}=\alpha _{N+l}$ for $k,l\in \{1,\ldots ,N\}$.}%
\begin{equation*}
\vartheta _{q}\doteq \beta ^{-1}(\tilde{\alpha}_{\pi ^{-1}(q)}-\tilde{\alpha}%
_{\pi ^{-1}(q-1)})\geq 0\ ,\text{ \ }\alpha _{\pi ^{-1}(q)}-\alpha _{\pi
^{-1}(q-1)}\geq 0\ ,\text{ \ }q\in \left\{ 2,\ldots ,2N\right\} \ ,
\end{equation*}%
where $\tilde{\alpha}_{q}\doteq \alpha _{q}$ for $q\in \{1,\ldots ,N\}$ and $%
\tilde{\alpha}_{q}\doteq \alpha _{q}+n^{-1}\beta $ for $q\in \{N+1,\ldots
,2N\}$, the following assertion holds true:
\begin{align}
& \mathrm{det}\left[ \mathfrak{M}_{j_{k},j_{N+l}}\left\langle \varphi
_{N+l},\left( C_{H}\kappa (\hat{H})\hat{\varphi}_{k}\right) \left( \alpha
_{k}-\alpha _{N+l}\right) \right\rangle _{\mathfrak{h}}\right] _{k,l=1}^{N}
\label{eq chiot ultime} \\
& =\left( -1\right) ^{\pi }\lim_{\nu \rightarrow \infty }\left\langle \Delta
_{\nu }^{\frac{1}{2}-\beta ^{-1}\tilde{\alpha}_{\pi
^{-1}(p-1)}}x_{p-1}^{\ast }\Delta _{\nu }^{\vartheta _{p-1}}\cdots
x_{2}^{\ast }\Delta _{\nu }^{\vartheta _{2}}x_{1}^{\ast }\eta _{\nu }\
,\right.  \notag \\
& \qquad \qquad \qquad \qquad \left. \Delta _{\nu }^{\beta ^{-1}\tilde{\alpha%
}_{\pi ^{-1}(p)}-\frac{1}{2}}x_{p}\Delta _{\nu }^{\vartheta
_{p+1}}x_{p+1}\cdots \Delta _{\nu }^{\vartheta _{2N}}x_{2N}\eta _{\nu
}\right\rangle _{\mathfrak{H}_{\nu }}\ .  \notag
\end{align}%
The integer $p$ is defined to be the smallest element of $\left\{ 1,\ldots
,2N\right\} $ so that $\tilde{\alpha}_{\pi (p)}\geq \beta /2$. $\Delta _{\nu
}$ is the modular operator associated with the pair $(\mathcal{X}_{\nu
},\eta _{\nu })$. For $q\in \{1,\ldots ,2N\}$ such that $\pi ^{-1}(q)\in
\{1,\ldots ,N\}$,%
\begin{equation*}
x_{q}\doteq \varkappa _{\nu }\left( a^{+}\left( \left( \mathrm{sgn}\left(
1-n^{-1}\beta H\right) ^{-\beta ^{-1}n\tilde{\alpha}_{\pi ^{-1}(q)}}\sqrt{%
\kappa \left( H\right) }\varphi _{\pi ^{-1}(q)}\right) \otimes \mathfrak{e}%
_{j_{\pi ^{-1}(q)}}\right) \right) \ ,
\end{equation*}%
while for $q\in \{1,\ldots ,2N\}$ such that $\pi ^{-1}(q)\in \{N+1,\ldots
,2N\}$,%
\begin{equation*}
x_{q}\doteq \varkappa _{\nu }\left( a\left( \left( \mathrm{sgn}\left(
1-n^{-1}\beta H\right) ^{\beta ^{-1}n\tilde{\alpha}_{\pi ^{-1}(q)}}\sqrt{%
\kappa \left( H\right) }\varphi _{\pi ^{-1}(q)}\right) \otimes \mathfrak{e}%
_{j_{\pi ^{-1}(q)}}\right) \right) \ .
\end{equation*}%
Here, $\mathrm{sgn}$ is the sign function defined as follows: $\mathrm{sgn}%
\left( x\right) \doteq 1$ for $x\in \mathbb{R}_{0}^{+}$ and $\mathrm{sgn}%
\left( x\right) \doteq -1$ otherwise.
\end{satz}

\begin{proof}
Combining Lemma \ref{lemma exp copy(2)} and Corollary \ref{Corollary chiot1}
with the construction done in Section \ref{Sectino chiot}, in particular
Equation (\ref{chiot final1}), one gets the assertion when all functions $%
\varphi _{1},\ldots ,\varphi _{N}\in \mathfrak{D}\subset \mathfrak{h}$
belong the dense space (\ref{fract domain}). To extend it to all $\varphi
_{1},\ldots ,\varphi _{N}\in \mathfrak{h}$, by (\ref{property araki chiot 1}%
), note that both sides of Equation (\ref{eq chiot ultime}) are continuous
with respect to $\varphi _{1},\ldots ,\varphi _{N}$.

For an explicit description of $(\mathfrak{H}_{\nu },\varkappa _{\nu },\eta
_{\nu })$, which is a cyclic\ representation of the quasi--free state $\rho
_{S_{\nu }}$ for $\nu \in \mathbb{R}^{+}$, see Sections \ref{Sect
Tree--Expansions} and \ref{Modular}. Heuristic arguments can be found in\
Section \ref{Bernoulli}.
\end{proof}

\begin{koro}[Determinant bounds]
\label{theorem main super2}\mbox{
}\newline
Under the assumptions of Theorem \ref{theorem main super},%
\begin{multline*}
\left\vert \mathrm{det}\left[ \mathfrak{M}_{j_{k},j_{N+l}}\left\langle
\varphi _{N+l},\left( C_{H}\kappa (\hat{H})\hat{\varphi}_{k}\right) \left(
\alpha _{k}-\alpha _{N+l}\right) \right\rangle _{\mathfrak{h}}\right]
_{k,l=1}^{N}\right\vert \\
\leq \prod_{q=1}^{2N}\left\Vert \sqrt{\kappa \left( H\right) }\varphi
_{q}\right\Vert _{\mathfrak{h}}\mathfrak{M}_{j_{q},j_{q}}^{1/2}\ .
\end{multline*}%
Compare with Definition \ref{determinant bounds}.
\end{koro}

\begin{proof}
This corollary is a direct consequence of Theorem \ref{theorem main super}
and Inequality (\ref{property araki chiot 1}). In fact, inequalities of the
form \cite[(A.2)]{Araki-Masuda} (which generalize (\ref{property araki chiot
1})) are intimately related to H\"{o}lder inequalities for non--commutative $%
L^{p}$--spaces. In the finite dimensional case, the non--commutative $L^{p}$%
--spaces correspond to spaces of Schatten class operators, as explained in
Section \ref{sect Schatten}.
\end{proof}

\begin{koro}[Universal determinant bounds]
\label{theorem main super2 copy(1)}\mbox{
}\newline
The universal determinant bound defined by (\ref{universal determinant bound}%
) equals $\mathfrak{x}=1$.
\end{koro}

\begin{proof}
Invoking Corollary \ref{theorem main super2}, we deduce $\mathfrak{x}\leq 1$%
, see (\ref{universal determinant bound}) and Definition \ref{determinant
bounds}. Now, let $\mathfrak{h}=\ell ^{2}(\mathbb{N};\mathbb{C})$ with
canonical ONB denoted by $\{e_{\mathfrak{i}}\}_{\mathfrak{i}\in \mathbb{N}}$%
. Take $\beta \in \mathbb{R}^{+}$, $\kappa =\mathbf{1}_{\mathbb{R}}$, $%
H=\lambda \mathbf{1}_{\mathfrak{h}}$ with $\lambda \in \mathbb{R}$ and $%
\mathfrak{M}\in \mathrm{Mat}\left( 1,\mathbb{R}\right) $ with $\mathfrak{M}%
_{1,1}=1$. Then, from Corollary \ref{Corollary chiot1} together with (\ref%
{symbol}) and (\ref{ass O0-00bis})--(\ref{ass O0-00bisbis}) , for each $n\in
2${$\mathbb{N}$} and all $N\in \mathbb{N}$, we directly compute that, for
sufficiently large $n\gg 1$,
\begin{equation*}
\left\vert \mathrm{det}\left[ \left\langle e_{k},\left( C_{\lambda \mathbf{1}%
_{\mathfrak{h}}}\hat{e}_{l}\right) \left( 0\right) \right\rangle _{\mathfrak{%
h}}\right] _{k,l=1}^{N}\right\vert =\left( 1-n^{-1}\beta \lambda \right)
^{-N}\left( 1+\left\vert 1-n^{-1}\beta \lambda \right\vert ^{-n}\right)
^{-N}.
\end{equation*}%
In particular, for every $\varepsilon >0$ and $\beta \in \mathbb{R}^{+}$,
there are $\lambda _{\varepsilon ,\beta }\in \mathbb{R}$ and $n_{\varepsilon
,\beta }\in \mathbb{N}$ such that, for all $n\geq n_{\varepsilon ,\beta }$
and $N\in \mathbb{N}$,
\begin{equation*}
\left\vert \mathrm{det}\left[ \left\langle e_{k},\left( C_{\lambda
_{\varepsilon ,\beta }\mathbf{1}_{\mathfrak{h}}}\hat{e}_{l}\right) \left(
0\right) \right\rangle _{\mathfrak{h}}\right] _{k,l=1}^{N}\right\vert \geq
\left( 1-\varepsilon \right) ^{2N}\ .
\end{equation*}%
Using this lower bound and Corollary \ref{theorem main super2}, we then
arrive at the equality $\mathfrak{x}=1$.
\end{proof}

\subsection{Finite Dimensional Case and H\"{o}lder Inequalities for Schatten
Norms\label{sect Schatten}}

As already discussed, we use H\"{o}lder inequalities for non--commutative $%
L^{p}$--spaces to derive determinant bounds (Definition \ref{determinant
bounds}). Here, we illustrate this approach in the finite dimensional case
via H\"{o}lder inequalities for Schatten norms: \medskip

\noindent \underline{(i):} Assume that $\mathfrak{h}$ is a \emph{finite
dimensional} Hilbert space. Then, the $C^{\ast }$--algebra $\mathrm{CAR}(%
\mathfrak{h}\otimes \mathbb{M})$ associated with $\mathfrak{h}\otimes
\mathbb{M}$ can be identified with the space $\mathcal{B}(\mathbb{F})$ of
all linear operators acting on the fermionic Fock space
\begin{equation*}
\mathbb{F}\doteq \wedge \left( \mathfrak{h}\otimes \mathbb{M}\right)
\end{equation*}%
constructed from the one--particle Hilbert space $\mathfrak{h}\otimes
\mathbb{M}$. \medskip

\noindent \underline{(ii):} Take any faithful state $\rho $ on $\mathcal{B}(%
\mathbb{F})$ with cyclic representation $(\mathfrak{H},\varkappa ,\eta )$.
By finite dimensionality, it follows that
\begin{equation*}
\varkappa \left( \mathrm{CAR}\left( \mathfrak{h}\otimes \mathbb{M}\right)
\right) ^{\prime \prime }=\varkappa \left( \mathrm{CAR}\left( \mathfrak{h}%
\otimes \mathbb{M}\right) \right) \ .
\end{equation*}%
Because $\rho $ is faithful and $\mathcal{B}(\mathbb{F})$ is a matrix
algebra, $\eta $ is separating for $\varkappa \left( \mathrm{CAR}\left(
\mathfrak{h}\otimes \mathbb{M}\right) \right) $ and the (Tomita--Takesaki)
modular objects associated with it are well--defined. Denote by $\Delta \in
\mathcal{B}(\mathfrak{H})$ the modular operator associated with the pair $%
\left( \varkappa \left( \mathrm{CAR}\left( \mathfrak{h}\otimes \mathbb{M}%
\right) \right) ,\eta \right) $. See Section \ref{Modular}.

The cyclic representation $(\mathfrak{H},\varkappa ,\eta )$ is uniquely
defined, up to a unitary transformation. It is explicitly given, for
instance, by the so--called standard (cyclic) representation \cite[Section
5.4]{DerezinskiFruboes2006}: The space $\mathfrak{H}$ corresponds to the
linear space $\mathcal{B}(\mathbb{F})$ endowed with the Hilbert--Schmidt
scalar product
\begin{equation}
\left\langle A,B\right\rangle _{\mathfrak{H}}\doteq \mathrm{Tr}_{\mathbb{F}%
}(A^{\ast }B)\ ,\qquad A,B\in \mathfrak{H}\ .  \label{trace h at}
\end{equation}%
For any $A\in \mathcal{B}(\mathbb{F})$ we define the left and right
multiplication operators $\underrightarrow{A}$ and $\underleftarrow{A}$
acting on $\mathcal{B}(\mathbb{F})$ by
\begin{equation*}
B\mapsto \underrightarrow{A}B\doteq AB\qquad \text{and}\qquad B\mapsto
\underleftarrow{A}B\doteq BA\ ,
\end{equation*}%
respectively. The representation $\varkappa $ is the left multiplication,
i.e.,
\begin{equation*}
\varkappa \left( A\right) \doteq \underrightarrow{A}\ ,\qquad A\in \mathcal{B%
}(\mathbb{F})\ .
\end{equation*}%
The cyclic vector $\eta $ is defined by
\begin{equation*}
\eta \doteq \mathrm{D}^{1/2}\in \mathfrak{H}
\end{equation*}%
with $\mathrm{D}\in \mathcal{B}(\mathbb{F})$ being the unique positive
operator such that
\begin{equation}
\rho \left( A\right) \doteq \mathrm{Tr}_{\mathbb{F}}\left( \mathrm{D}%
A\right) \ ,\qquad A\in \mathcal{B}(\mathbb{F})\ .  \label{eq trace}
\end{equation}%
I.e., $\mathrm{D}$ is the density matrix of the state $\rho $. In this
representation, the modular operator $\Delta $ associated with $\rho $ is
equal to
\begin{equation}
\Delta =\underrightarrow{\mathrm{D}}\ \underleftarrow{\mathrm{D}^{-1}}\in
\mathcal{B}(\mathfrak{H})\ .  \label{modular}
\end{equation}%
Note that if a state is faithful then its density matrix $\mathrm{D}$ is
invertible. The ($\beta $--rescaled) modular group is the one--parameter
group $\sigma \equiv \{\sigma _{t}\}_{t\in {\mathbb{R}}}$ defined by
\begin{equation}
\sigma _{t}(\underrightarrow{A})\doteq \Delta ^{-it\beta ^{-1}}%
\underrightarrow{A}\Delta ^{it\beta ^{-1}}\qquad A\in \mathcal{B}(\mathbb{F}%
)\ .  \label{assertion chiot bresiliennes bis}
\end{equation}%
\medskip

\noindent \underline{(iii):} Now, we fix $n\in 2${$\mathbb{N}$} and apply
this last construction to the quasi--free states $\rho =\rho _{S_{\nu }}$, $%
\nu \in \mathbb{R}^{+}$, which are defined from symbols $S_{\nu }$ (\ref%
{symbol}). See Section \ref{Sect Tree--Expansions}. Denote their standard
representations by $(\mathfrak{H}_{\nu },\varkappa _{\nu },\eta _{\nu })$,
their density matrices by $\mathrm{D}_{\nu }$ and the associated modular
operators by $\Delta _{\nu }$. We infer from (\ref{trace h at}), (\ref{eq
trace}), (\ref{modular}), Corollary \ref{Corollary chiot1}, the defining
properties of Bogoliubov automorphisms (compare (\ref{assertion chiot
bresiliennes bis}) with (\ref{definition tho})--(\ref{assertion chiot
bresiliennes})), the cyclicity of traces, and the assumptions and
definitions of Theorem \ref{theorem main super} that%
\begin{align}
& \mathrm{det}\left[ \mathfrak{M}_{j_{k},j_{N+l}}\left\langle \varphi
_{N+l},\left( C_{H}\kappa (\hat{H})\hat{\varphi}_{k}\right) \left( \alpha
_{k}-\alpha _{N+l}\right) \right\rangle _{\mathfrak{h}}\right] _{k,l=1}^{N}
\notag \\
& =\lim_{\nu \rightarrow \infty }\left( -1\right) ^{\pi }\mathrm{Tr}_{%
\mathbb{F}}\left( \mathrm{D}_{\nu }^{\tilde{\alpha}_{\pi ^{-1}(1)}\beta
^{-1}}\mathrm{D}_{\nu }^{\frac{1}{2}}x_{1}\left( \prod_{j=2}^{p-1}\left(
\mathrm{D}_{\nu }^{\vartheta _{j}}x_{j}\right) \right) \mathrm{D}_{\nu }^{%
\frac{1}{2}-\beta ^{-1}\tilde{\alpha}_{\pi ^{-1}(p-1)}}\right.  \notag \\
& \qquad \qquad \qquad \left. \mathrm{D}_{\nu }^{\beta ^{-1}\tilde{\alpha}%
_{\pi ^{-1}(p)}-\frac{1}{2}}x_{p}\left( \prod_{j=p+1}^{2N}\left( \mathrm{D}%
_{\nu }^{\vartheta _{j}}x_{j}\right) \right) \mathrm{D}_{\nu }^{\frac{1}{2}}%
\mathrm{D}_{\nu }^{-\beta ^{-1}\tilde{\alpha}_{\pi ^{-1}(2N)}}\right)  \notag
\\
& =\lim_{\nu \rightarrow \infty }\left( -1\right) ^{\pi }\left\langle \Delta
_{\nu }^{\frac{1}{2}-\beta ^{-1}\tilde{\alpha}_{\pi
^{-1}(p-1)}}x_{p-1}^{\ast }\Delta _{\nu }^{\vartheta _{p-1}}\cdots
x_{2}^{\ast }\Delta _{\nu }^{\vartheta _{2}}x_{1}^{\ast }\eta _{\nu }\
,\right.  \label{chiot ultime2} \\
& \qquad \qquad \qquad \qquad \left. \Delta _{\nu }^{\beta ^{-1}\tilde{\alpha%
}_{\pi ^{-1}(p)}-\frac{1}{2}}x_{p}\Delta _{\nu }^{\vartheta
_{p+1}}x_{p+1}\cdots \Delta _{\nu }^{\vartheta _{2N}}x_{2N}\eta _{\nu
}\right\rangle _{\mathfrak{H}_{\nu }},  \notag
\end{align}%
that is, Equation (\ref{eq chiot ultime}). \medskip

\noindent \underline{(iv):} Schatten norms on $\mathcal{B}(\mathbb{F})$\ are
defined by
\begin{equation*}
\left\Vert A\right\Vert _{s}\doteq \left( \mathrm{Tr}_{\mathbb{F}}\left(
\left\vert A\right\vert ^{s}\right) \right) ^{\frac{1}{s}}\ ,\qquad A\in
\mathcal{B}(\mathbb{F})\ ,\ s\geq 1\ ,
\end{equation*}%
and
\begin{equation*}
\left\Vert A\right\Vert _{\infty }\doteq \lim_{s\rightarrow \infty }\left(
\mathrm{Tr}_{\mathbb{F}}\left( \left\vert A\right\vert ^{s}\right) \right) ^{%
\frac{1}{s}}=\left\Vert A\right\Vert _{\mathcal{B}(\mathbb{F})}\ ,\qquad
A\in \mathcal{B}(\mathbb{F})\ .
\end{equation*}%
Remark that the norm on the Hilbert space $\mathfrak{H}$ defined from the
scalar product (\ref{trace h at}) is the Hilbert--Schmidt norm, i.e.,
\begin{equation}
\left\Vert A\right\Vert _{\mathfrak{H}}=\left\Vert A\right\Vert _{2}\
,\qquad A\in \mathcal{B}(\mathbb{F})\equiv \mathfrak{H}\ .
\label{Holder jean sans bras0}
\end{equation}%
\medskip

\noindent \underline{(v):} H\"{o}lder inequalities for Schatten norms refer
to the following bounds: For any $n\in 2${$\mathbb{N}$}, $r,s_{1},\ldots
,s_{n}\in \lbrack 1,\infty ]$ such that $\sum_{j=1}^{n}1/s_{j}=1/r$, and all
operators $A_{1},\ldots ,A_{n}\in \mathcal{B}(\mathbb{F})$,
\begin{equation}
\left\Vert A_{1}\cdots A_{n}\right\Vert _{r}\leq \prod_{j=1}^{n}\left\Vert
A_{j}\right\Vert _{s_{j}}\ .  \label{Holder jean sans bras}
\end{equation}%
This type of inequality combined with (\ref{chiot ultime2}) implies
Corollary \ref{theorem main super2} in the finite dimensional case. \medskip

\noindent \underline{(vi):} Indeed, for any integer $N\in \mathbb{N}$ and
strictly positive parameter $\zeta \in \mathbb{R}^{+}$, define the tube%
\begin{equation}
\mathfrak{T}_{N}^{(\zeta )}\doteq \left\{ \left( z_{1},\ldots ,z_{N}\right)
\in \mathbb{C}^{N}:\forall j\in \{1,\ldots ,N\},\ \mathrm{Re}(z_{j})\geq 0,%
\text{ }\sum_{j=1}^{N}\mathrm{Re}(z_{j})\leq \zeta \right\} \ .  \label{tube}
\end{equation}%
Let $\rho $ be a faithful quasi--free state on $\mathcal{B}(\mathbb{F})$ and
denote by $H_{\rho }=H_{\rho }^{\ast }\in \mathcal{B}(\mathfrak{h}\otimes
\mathbb{M})$ the unique self--adjoint operator such that the symbol $%
S^{(\rho )}$ of $\rho $ equals
\begin{equation*}
S^{(\rho )}=\frac{1}{1+\mathrm{e}^{H_{\rho }}}\ .
\end{equation*}%
See beginning of Section \ref{Sect Tree--Expansions} for more explanations
on quasi--free states in relation with their symbols.

Choose $\Psi _{1},\ldots ,\Psi _{N}\in \mathfrak{h}\otimes \mathbb{M}$ and
pick a family$\{a^{\#}\left( \Psi _{q}\right) \}_{q=1}^{N}$ of elements of $%
\mathrm{CAR}(\mathfrak{h}\otimes \mathbb{M})$, where the notation
\textquotedblleft $a^{\#}$\textquotedblright\ stands for either
\textquotedblleft $a^{+}$\textquotedblright\ or \textquotedblleft $a$%
\textquotedblright . For any complex vector $\left( z_{1},\ldots
,z_{N}\right) \in \mathfrak{T}_{N}^{(1/2)}$, we observe from (\ref{modular})
that
\begin{eqnarray*}
&&\Delta ^{z_{1}}\varkappa \left( a^{\#}\left( \Psi _{1}\right) \right)
\Delta ^{z_{2}}\cdots \Delta ^{z_{N}}\varkappa \left( a^{\#}\left( \Psi
_{N}\right) \right) \eta \\
&=&\mathrm{D}^{\mathrm{Re}\left( z_{1}\right) }a^{\#}\left( \mathrm{e}^{-i%
\mathrm{Im}\left( z_{1}\right) H_{\rho }}\Psi _{1}\right) \mathrm{D}^{%
\mathrm{Re}\left( z_{2}\right) }a^{\#}\left( \mathrm{e}^{-i\left( \mathrm{Im}%
\left( z_{1}\right) +\mathrm{Im}\left( z_{2}\right) \right) H_{\rho }}\Psi
_{2}\right) \\
&&\cdots \mathrm{D}^{\mathrm{Re}\left( z_{N}\right) }a^{\#}\left( \mathrm{e}%
^{-i\left( \mathrm{Im}\left( z_{1}\right) +\cdots +\mathrm{Im}\left(
z_{N}\right) \right) H_{\rho }}\Psi _{N}\right) \mathrm{D}^{1/2-\left(
\mathrm{Re}\left( z_{1}\right) +\cdots +\mathrm{Re}\left( z_{N}\right)
\right) }\ .
\end{eqnarray*}%
By applying H\"{o}lder inequalities (\ref{Holder jean sans bras0}) and (\ref%
{Holder jean sans bras}), we obtain from the last equality that%
\begin{align}
& \left\Vert \Delta ^{z_{1}}\varkappa \left( a^{\#}\left( \Psi _{1}\right)
\right) \Delta ^{z_{2}}\cdots \Delta ^{z_{N}}\varkappa \left( a^{\#}\left(
\Psi _{N}\right) \right) \eta \right\Vert _{\mathfrak{H}}  \notag \\
& \leq \left\Vert \mathrm{D}^{1/2-\left( \mathrm{Re}\left( z_{1}\right)
+\cdots +\mathrm{Re}\left( z_{N}\right) \right) }\right\Vert _{\frac{1}{%
1/2-\left( \mathrm{Re}\left( z_{1}\right) +\cdots +\mathrm{Re}\left(
z_{N}\right) \right) }}  \notag \\
& \qquad \qquad \qquad \qquad \times \prod_{q=1}^{N}\left\Vert \mathrm{D}^{%
\mathrm{Re}\left( z_{q}\right) }\right\Vert _{\frac{1}{\mathrm{Re}\left(
z_{q}\right) }}\left\Vert a^{\#}\left( \mathrm{e}^{i\left( \mathrm{Im}\left(
z_{1}\right) +\cdots +\mathrm{Im}\left( z_{q}\right) \right) H_{\rho }}\Psi
_{q}\right) \right\Vert _{\infty }\ ,  \notag
\end{align}%
which, combined with $\left\Vert \mathrm{D}\right\Vert _{1}=1$, in turn
implies that%
\begin{equation}
\left\Vert \Delta ^{z_{1}}\varkappa \left( a^{\#}\left( \Psi _{1}\right)
\right) \Delta ^{z_{2}}\cdots \Delta ^{z_{N}}\varkappa \left( a^{\#}\left(
\Psi _{N}\right) \right) \eta \right\Vert _{\mathfrak{H}}\leq
\prod_{q=1}^{N}\left\Vert \Psi _{q}\right\Vert _{\mathfrak{h}\otimes \mathbb{%
M}}\ .  \label{Holder finite dimHolder finite dim}
\end{equation}

This inequality corresponds to (\ref{property araki chiot 1}) in the finite
dimensional case. Therefore, Equation (\ref{chiot ultime2}) combined with
Inequality (\ref{Holder finite dimHolder finite dim}) implies Corollary \ref%
{theorem main super2} when $\mathfrak{h}$ is a finite dimensional Hilbert
space.

\section{Technical Proofs\label{Thecnical proofs}}

\subsection{Quasi--Free States on General Monomials\label{Quasi free sect}}

Let $\mathcal{H}$ be some Hilbert space and $\mathrm{CAR}(\mathcal{H})$ the
associated CAR $C^{\ast }$--algebra generated by the unit $\mathbf{1}$ and
the family $\{a(\varphi )\}_{\varphi \in \mathcal{H}}$ of elements
satisfying the canonical commutation relations (CAR): For any $\varphi
_{1},\varphi _{2}\in \mathcal{H}$,
\begin{eqnarray}
a(\varphi _{1})a(\varphi _{2})+a(\varphi _{2})a(\varphi _{1}) &=&0\text{ },
\label{CAR AA} \\
a(\varphi _{1})^{\ast }a(\varphi _{2})+a(\varphi _{2})a(\varphi _{1})^{\ast
} &=&\left\langle \varphi _{2},\varphi _{1}\right\rangle _{\mathcal{H}}%
\mathbf{1}\text{ }.  \label{CAR AA*}
\end{eqnarray}%
Strictly speaking, the above conditions only define $\mathrm{CAR}(\mathcal{H}%
)$ up to an isomorphism of $C^{\ast }$--algebras. See, e.g., \cite[Theorem
5.2.5]{BratteliRobinson}. As explained in Section \ref{Sect Tree--Expansions}
for the special case $\mathcal{H}=\mathfrak{h}\otimes \mathbb{M}$, the
generator $a(\varphi )\in \mathrm{CAR}(\mathcal{H})$ is interpreted as the
annihilation operator associated with $\varphi \in \mathcal{H}$ whereas its
adjoint
\begin{equation*}
a^{+}(\varphi )\doteq a(\varphi )^{\ast }\ ,\qquad \varphi \in \mathcal{H}\ ,
\end{equation*}%
is the corresponding creation operator.

A monomial in the annihilation and creation operators is \emph{normally
ordered} if the creation operators appearing in the monomial are on the left
side of all annihilation operators in the same monomial, like
\begin{equation*}
a^{+}(\varphi _{1})\cdots a^{+}(\varphi _{N_{1}})a(\varphi _{2N_{1}})\cdots
a(\varphi _{N_{1}+1})\ .
\end{equation*}%
By the above definition, if $\rho $ is a quasi--free state and $\mathcal{M}%
\in \mathrm{CAR}(\mathcal{H})$ is a normally ordered monomial in the
annihilation and creation operators, then $\rho (\mathcal{M})$ is the
determinant of a matrix, the entries of which are given by $\rho $ acting on
monomials of degree two. We show below that this pivotal property of
quasi--free states remains valid even if $\mathcal{M}$ is \emph{not}
normally ordered.

This is not surprising. For instance, \cite[Definition 3.1, Condition (3.2)]%
{Araki} also essentially says that if the state is quasi--free then
expectation values (with respect to this state) of any monomial (not
necessarily normally ordered) of arbitrary even degree is a determinant of a
matrix, the entries of which are expectation values of monomials of degree
two. However, beyond this fact, \ we would like to give the \emph{explicit}
behavior of such expectation values with respect to arbitrary permutations
of creation and annihilation operators in large monomials. This point is
crucial here and is given by Lemma \ref{lemma exp copy(2)}.

To this end, we introduce some notation. If $\pi $ is a permutation of $n\in
{\mathbb{N}}$ elements (i.e., a bijective function from $\{1,\ldots ,n\}$ to
$\{1,\ldots ,n\}$) with sign $(-1)^{\pi }$, we define the monomial $\mathbb{O%
}_{\pi }(A_{1},\ldots ,A_{n})\in \mathrm{CAR}(\mathcal{H})$ in $A_{1},\ldots
,A_{n}\in \mathrm{CAR}(\mathcal{H})$ by the product%
\begin{equation}
\mathbb{O}_{\pi }\left( A_{1},\ldots ,A_{n}\right) \doteq \left( -1\right)
^{\pi }A_{\pi ^{-1}(1)}\cdots A_{\pi ^{-1}(n)}\text{ }.  \label{Opi1}
\end{equation}%
In other words, $\mathbb{O}_{\pi }$ places the operator $A_{k}$ at the $\pi
(k)$th position in the monomial $(-1)^{\pi }A_{\pi ^{-1}(1)}\cdots A_{\pi
^{-1}(n)}$. Further, for all $k,l\in \{1,\ldots ,n\}$, $k\neq l$,%
\begin{equation}
\pi _{k,l}:\{1,2\}\rightarrow \{1,2\}  \label{Opi2}
\end{equation}%
is the identity function if $\pi (k)<\pi (l)$, otherwise $\pi _{k,l}$
interchanges $1$ and $2$. Then, the following property of quasi--free states
holds true:

\begin{lemma}[Quasi--free states on general monomials]
\label{lemma exp copy(2)}\mbox{ }\newline
Let $\rho $ be a quasi--free state on the $C^{\ast }$--algebra $\mathrm{CAR}(%
\mathcal{H})$, as defined by (\ref{ass O0-00})--(\ref{ass O0-00bis}) for $%
\mathcal{H}=\mathfrak{h}\otimes \mathbb{M}$. For any $N_{1},N_{2}\in \mathbb{%
N}$, all permutations $\pi $ of $N_{1}+N_{2}$ elements and $\varphi
_{1},\ldots ,\varphi _{N_{1}+N_{2}}\in \mathcal{H}$,%
\begin{equation}
\rho
\Big(%
\mathbb{O}_{\pi }\left( a^{+}(\varphi _{1}),\ldots ,a^{+}(\varphi
_{N_{1}}),a(\varphi _{N_{1}+N_{2}}),\ldots ,a(\varphi _{N_{1}+1})\right)
\Big)%
=0  \label{ass O0}
\end{equation}%
if $N_{1}\neq N_{2}$, while in the case $N_{1}=N_{2}\equiv N$,%
\begin{eqnarray}
&&\rho
\Big(%
\mathbb{O}_{\pi }\left( a^{+}(\varphi _{1}),\ldots ,a^{+}(\varphi
_{N}),a(\varphi _{2N}),\ldots ,a(\varphi _{N+1})\right)
\Big)
\notag \\
&=&\mathrm{det}\left[ \rho
\Big(%
\mathbb{O}_{\pi _{k,N+l}}\left( a^{+}(\varphi _{k}),a(\varphi _{N+l})\right)
\Big)%
\right] _{k,l=1}^{N}\text{ }.  \label{ass det O}
\end{eqnarray}
\end{lemma}

\begin{proof}
By (\ref{CAR AA}) and (\ref{CAR AA*}), if the monomial
\begin{equation*}
\mathbb{O}_{\pi }\left( a^{+}(\varphi _{1}),\ldots ,a^{+}(\varphi
_{N_{1}}),a(\varphi _{N_{1}+N_{2}}),\ldots ,a(\varphi _{N_{1}+1})\right)
\end{equation*}%
contains different numbers of annihilation and creation operators (i.e., $%
N_{1}\neq N_{2}$), then it can be written as a sum of normally ordered
monomials with the same property. By (\ref{ass O0-00}) and the linearity of
states, we thus deduce (\ref{ass O0}).

We consider the case $N_{1}=N_{2}\equiv N\in \mathbb{N}$. Assertion (\ref%
{ass det O}) trivially holds if $N=1$ and we can assume from now on that $%
N\geq 2$.

For convenience, the notation \textquotedblleft $a^{\#}$\textquotedblright\
stands for either \textquotedblleft $a^{+}$\textquotedblright\ or
\textquotedblleft $a$\textquotedblright . In particular, we write the
monomial%
\begin{equation*}
\mathbb{O}_{\pi }\left( a^{+}(\varphi _{1}),\ldots ,a^{+}(\varphi
_{n}),a(\varphi _{2N}),\ldots ,a(\varphi _{N+1})\right) =\left( -1\right)
^{\pi }a_{1}^{\#}\cdots a_{2N}^{\#}\text{ }.
\end{equation*}%
Let%
\begin{eqnarray*}
k_{\pi } &\doteq &\min \left\{ \pi \left( N+1\right) ,\ldots ,\pi \left(
2N\right) \right\} \leq N+1\ , \\
k_{\pi }^{+} &\doteq &\max \left\{ \pi \left( 1\right) ,\ldots ,\pi \left(
N\right) \right\} \geq N\ .
\end{eqnarray*}%
The parameter $k_{\pi }$ is the position the first annihilation operator
appearing in the monomial $a_{1}^{\#}\cdots a_{2N}^{\#}$ while $k_{\pi }^{+}$
is the position of the last creation operator appearing in $a_{1}^{\#}\cdots
a_{2N}^{\#}$. In other words,%
\begin{eqnarray*}
&&\mathbb{O}_{\pi }\left( a^{+}(\varphi _{1}),\ldots ,a^{+}(\varphi
_{N}),a(\varphi _{2N}),\ldots ,a(\varphi _{N+1})\right) \\
&=&(-1)^{\pi }a_{1}^{+}\cdots a_{k_{\pi }-1}^{+}a_{k_{\pi }}a_{k_{\pi
}+1}^{\#}\cdots a_{k_{\pi }^{+}-1}^{\#}a_{k_{\pi }^{+}}^{+}a_{k_{\pi
}^{+}+1}\cdots a_{2N}\text{ }.
\end{eqnarray*}%
Note that $k_{\pi }=N+1$ iff the monomial is normally ordered. The same
holds true if $k_{\pi }^{+}=N$. In particular, $k_{\pi }=N+1$ iff $k_{\pi
}^{+}=N$. We will prove Assertion (\ref{ass det O}) by induction in the
parameter
\begin{equation*}
N_{\pi }\doteq k_{\pi }^{+}-k_{\pi }+1\geq 0\ .
\end{equation*}%
Observe that $N_{\pi }=0$ iff the monomial is normally ordered and Assertion
(\ref{ass det O}) holds in this case because of (\ref{CAR AA}), (\ref{ass
O0-00bis}) and the antisymmetry of the determinant under permutations of its
lines or rows.

Assume now that $N_{\pi }\geq 1$. Thus, $k_{\pi }\leq N$ and $k_{\pi
}^{+}\geq N+1$. If $k_{\pi }>2$ and $2N-k_{\pi }>3$ then%
\begin{eqnarray}
&&\left( -1\right) ^{\pi }\mathbb{O}_{\pi }\left( a^{+}(\varphi _{1}),\ldots
,a^{+}(\varphi _{N}),a(\varphi _{2N}),\ldots ,a(\varphi _{N+1})\right)
\notag \\
&=&a_{1}^{+}\cdots a_{k_{\pi }-1}^{+}\{a_{k_{\pi }},a_{k_{\pi
}+1}^{\#}\}a_{k_{\pi }+2}^{\#}\cdots a_{2N}^{\#}  \label{eq sup z1} \\
&&-a_{1}^{+}\cdots a_{k_{\pi }-1}^{+}a_{k_{\pi }+1}^{\#}\{a_{k_{\pi
}},a_{k_{\pi }+2}^{\#}\}a_{k_{\pi }+3}^{\#}\cdots a_{2N}^{\#}  \notag \\
&&+a_{1}^{+}\cdots a_{k_{\pi }-1}^{+}\sum_{l=3}^{2N-k_{\pi
}-2}(-1)^{l-1}a_{k_{\pi }+1}^{\#}\cdots a_{k_{\pi }+l-1}^{\#}\{a_{k_{\pi
}},a_{k_{\pi }+l}^{\#}\}a_{k_{\pi }+l+1}^{\#}\cdots a_{2N}^{\#}  \notag \\
&&+(-1)^{2N-k_{\pi }-2}a_{1}^{+}\cdots a_{k_{\pi }-1}^{+}a_{k_{\pi
}+1}^{\#}\cdots a_{2N-2}^{\#}\{a_{k_{\pi }},a_{2N-1}^{\#}\}a_{2N}^{\#}
\notag \\
&&+(-1)^{2N-k_{\pi }-1}a_{1}^{+}\cdots a_{k_{\pi }-1}^{+}a_{k_{\pi
}+1}^{\#}\cdots a_{2N-1}^{\#}\{a_{k_{\pi }},a_{2N}^{\#}\}  \notag \\
&&+(-1)^{2N-k_{\pi }}a_{1}^{+}\cdots a_{k_{\pi }-1}^{+}a_{k_{\pi
}+1}^{\#}\cdots a_{2N}^{\#}a_{k_{\pi }}\text{ },  \notag
\end{eqnarray}%
with $\{A,B\}\doteq AB+BA$. Mutatis mutandis if $k_{\pi }=1,2$ or $2N-k_{\pi
}=2,3$. It is convenient to use the definition
\begin{equation*}
q_{\pi }\doteq 2N-\pi ^{-1}(k_{\pi })+1\in \left\{ 1,\ldots ,N\right\} \ ,
\end{equation*}%
which implies $a_{k_{\pi }}=a(\varphi _{N+q_{\pi }})$. By combining (\ref{eq
sup z1}) with the CAR (\ref{CAR AA}) and (\ref{CAR AA*}), we deduce the
equality%
\begin{eqnarray}
&&\left( -1\right) ^{\pi }\mathbb{O}_{\pi }\left( a^{+}(\varphi _{1}),\ldots
,a^{+}(\varphi _{N}),a(\varphi _{2N}),\ldots ,a(\varphi _{N+1})\right)
\notag \\
&=&(-1)^{2N-k_{\pi }}a_{1}^{+}\cdots a_{k_{\pi }-1}^{+}a_{k_{\pi
}+1}^{\#}\cdots a_{2N}^{\#}a_{k_{\pi }}  \label{super eq} \\
&&+\sum_{k=1}^{N}\left\langle \varphi _{N+q_{\pi }},\varphi
_{k}\right\rangle _{\mathcal{H}}\left\{ \mathbf{1}\left[ k_{\pi }+1=\pi (k)%
\right] a_{1}^{+}\cdots a_{k_{\pi }-1}^{+}a_{k_{\pi }+2}^{\#}\cdots
a_{2N}^{\#}\right.  \notag \\
&&-\mathbf{1}\left[ k_{\pi }+2=\pi (k)\right] a_{1}^{+}\cdots a_{k_{\pi
}-1}^{+}a_{k_{\pi }+1}^{\#}a_{k_{\pi }+3}^{\#}\cdots a_{2N}^{\#}  \notag \\
&&+\mathbf{1}\left[ 2N-1>\pi (k)>k_{\pi }+2\right] (-1)^{\pi (k)-k_{\pi }-1}
\notag \\
&&\qquad \qquad \qquad \qquad \qquad \qquad a_{1}^{+}\cdots a_{k_{\pi
}-1}^{+}a_{k_{\pi }+1}^{\#}\cdots a_{\pi (k)-1}^{\#}a_{\pi (k)+1}^{\#}\cdots
a_{2N}^{\#}  \notag \\
&&+\mathbf{1}\left[ 2N-1=\pi (k)\right] (-1)^{2N-k_{\pi }-2}a_{1}^{+}\cdots
a_{k_{\pi }-1}^{+}a_{k_{\pi }+1}^{\#}\cdots a_{2N-2}^{\#}a_{2N}^{\#}  \notag
\\
&&\left. +\mathbf{1}\left[ 2N=\pi (k)\right] (-1)^{2N-k_{\pi
}-1}a_{1}^{+}\cdots a_{k_{\pi }-1}^{+}a_{k_{\pi }+1}^{\#}\cdots
a_{2N-1}^{\#}\right\}  \notag
\end{eqnarray}%
when $k_{\pi }>2$ and $2N-k_{\pi }>3$. Mutatis mutandis if $k_{\pi }=1,2$ or
$2N-k_{\pi }=2,3$. For any $k\in \{1,\ldots ,N\}$, we fix a permutation $\pi
^{(k)}$ of $2(N-1)$ elements such that%
\begin{multline*}
\mathbb{O}_{\pi ^{(k)}}%
\Big(%
a^{+}(\varphi _{1}),\ldots ,a^{+}(\varphi _{k-1}),a^{+}(\varphi
_{k+1}),\ldots ,a^{+}(\varphi _{N}), \\
a(\varphi _{2N}),\ldots ,a(\varphi _{N+q_{\pi }+1}),a(\varphi _{N+q_{\pi
}-1}),\ldots ,a(\varphi _{N+1})%
\Big)
\\
=\left( -1\right) ^{\pi ^{(k)}}a_{1}^{+}\cdots a_{k_{\pi }-1}^{+}a_{k_{\pi
}+1}^{\#}\cdots a_{\pi (k)-1}^{\#}a_{\pi (k)+1}^{\#}\cdots a_{2N}^{\#}\text{
}.
\end{multline*}%
(Recall that $N\geq 2$ is assumed without loss of generality.) Similarly, $%
\tilde{\pi}$ is a permutation of $2N$ elements such that%
\begin{multline*}
\mathbb{O}_{\tilde{\pi}}\left( a^{+}(\varphi _{1}),\ldots ,a^{+}(\varphi
_{N}),a(\varphi _{2N}),\ldots ,a(\varphi _{N+1})\right) \\
=\left( -1\right) ^{\tilde{\pi}}a_{1}^{+}\cdots a_{k_{\pi }-1}^{+}a_{k_{\pi
}+1}^{\#}\cdots a_{2N}^{\#}a_{k_{\pi }}\text{ }.
\end{multline*}%
By using this notation, we rewrite (\ref{super eq}) as%
\begin{eqnarray}
&&\left( -1\right) ^{\pi }\mathbb{O}_{\pi }\left( a^{+}(\varphi _{1}),\ldots
,a^{+}(\varphi _{N}),a(\varphi _{2N}),\ldots ,a(\varphi _{N+1})\right)
\notag \\
&=&(-1)^{k_{\pi }}\left( -1\right) ^{\tilde{\pi}}\mathbb{O}_{\tilde{\pi}%
}\left( a^{+}(\varphi _{1}),\ldots ,a^{+}(\varphi _{N}),a(\varphi
_{2N}),\ldots ,a(\varphi _{N+1})\right)  \label{super eqsuper eq} \\
&&+\sum_{k=1}^{N}\mathbf{1}\left[ \pi (k)>k_{\pi }\right] (-1)^{\pi
(k)-k_{\pi }-1}\left\langle \varphi _{N+q_{\pi }},\varphi _{k}\right\rangle
_{\mathcal{H}}\left( -1\right) ^{\pi ^{(k)}}  \notag \\
&&\times \mathbb{O}_{\pi ^{(k)}}%
\Big(%
a^{+}(\varphi _{1}),\ldots ,a^{+}(\varphi _{k-1}),a^{+}(\varphi
_{k+1}),\ldots ,a^{+}(\varphi _{N}),a(\varphi _{2N}),  \notag \\
&&\qquad \qquad \qquad \qquad \ldots ,a(\varphi _{N+q_{\pi }+1}),a(\varphi
_{N+q_{\pi }-1}),\ldots ,a(\varphi _{N+1})%
\Big)%
\ .  \notag
\end{eqnarray}%
For all $k\in \{1,\ldots ,N\}$, note that%
\begin{equation*}
k_{\pi ^{(k)}}^{+}\leq k_{\pi }^{+}-2\qquad \text{and}\qquad k_{\pi
^{(k)}}\geq k_{\pi }\ .
\end{equation*}%
As a consequence, for any $k\in \{1,\ldots ,N\}$, the induction parameter $%
N_{\pi ^{(k)}}$ associated with the permutation $\pi ^{(k)}$ satisfies:
\begin{equation}
N_{\pi ^{(k)}}\doteq k_{\pi ^{(k)}}^{+}-k_{\pi ^{(k)}}+1\leq k_{\pi
}^{+}-k_{\pi }-1=N_{\pi }-2\text{ }.  \label{inductin hypothesis cool1}
\end{equation}%
Similarly,
\begin{equation*}
k_{\tilde{\pi}}^{+}=k_{\pi }^{+}-1\qquad \text{and}\qquad k_{\tilde{\pi}%
}\geq k_{\pi }\ ,
\end{equation*}%
which in turn imply%
\begin{equation}
N_{\tilde{\pi}}\doteq k_{\tilde{\pi}}^{+}-k_{\tilde{\pi}}+1\leq k_{\pi
}^{+}-k_{\pi }=N_{\pi }-1\text{ }.  \label{inductin hypothesis cool2}
\end{equation}%
Observe furthermore that, for any $k\in \{1,\ldots ,N\}$ such that $\pi
(k)>k_{\pi }$,
\begin{equation}
(-1)^{\tilde{\pi}}=(-1)^{\pi }(-1)^{k_{\pi }}\qquad \text{and}\qquad
(-1)^{\pi ^{(k)}}=(-1)^{\pi }(-1)^{q_{\pi }+k+k_{\pi }+\pi (k)}\ .
\label{sign sympato}
\end{equation}%
Therefore, by using (\ref{sign sympato}) together with (\ref{super eqsuper
eq}), we arrive at the equality%
\begin{eqnarray}
&&\mathbb{O}_{\pi }\left( a^{+}(\varphi _{1}),\ldots ,a^{+}(\varphi
_{N}),a(\varphi _{2N}),\ldots ,a(\varphi _{N+1})\right)  \notag \\
&=&\mathbb{O}_{\tilde{\pi}}\left( a^{+}(\varphi _{1}),\ldots ,a^{+}(\varphi
_{N}),a(\varphi _{2N}),\ldots ,a(\varphi _{N+1})\right)
\label{super eqsuper eqsuper eqsuper eq} \\
&&-\sum_{k=1}^{N}\mathbf{1}\left[ \pi (k)>k_{\pi }\right] (-1)^{q_{\pi
}+k}\left\langle \varphi _{N+q_{\pi }},\varphi _{k}\right\rangle _{\mathcal{H%
}}  \notag \\
&&\qquad \qquad \mathbb{O}_{\pi ^{(k)}}%
\Big(%
a^{+}(\varphi _{1}),\ldots ,a^{+}(\varphi _{k-1}),a^{+}(\varphi
_{k+1}),\ldots ,a^{+}(\varphi _{N}),a(\varphi _{2N}),  \notag \\
&&\qquad \qquad \qquad \qquad \qquad \ \ \ \ldots ,a(\varphi _{N+q_{\pi
}+1}),a(\varphi _{N+q_{\pi }-1}),\ldots ,a(\varphi _{N+1})%
\Big)%
\ .  \notag
\end{eqnarray}%
We use now the following definitions: For any $k,l\in \{1,\ldots ,N\}$, the
coefficients%
\begin{eqnarray*}
\mathrm{M}_{k,l} &\doteq &\rho \left( \mathbb{O}_{\pi _{k,N+l}}\left(
a^{+}(\varphi _{k}),a(\varphi _{N+l})\right) \right) \ , \\
\mathrm{\tilde{M}}_{k,l} &\doteq &\rho \left( \mathbb{O}_{\tilde{\pi}%
_{k,N+l}}\left( a^{+}(\varphi _{k}),a(\varphi _{N+l})\right) \right) \ ,
\end{eqnarray*}%
$k,l\in \{1,\ldots ,N\}$, are the entries of two matrices $\mathrm{M}$ and $%
\mathrm{\tilde{M}}$, respectively. Let
\begin{equation*}
\mathrm{M}^{(k,l)}\doteq \det \left( \left[ \mathrm{M}_{i,j}\right] _{i,j\in
\{1,\ldots ,N\},i\neq k,j\neq l}\right)
\end{equation*}%
be the $k,l\in \{1,\ldots ,N\}\ $minor of $\mathrm{M}$, that is, the
determinant of the $(N-1)\times (N-1)$ matrix that results from deleting the
$k$th row and the $l$th column of $\mathrm{M}$. From the Laplace expansion
for determinants (sometimes called cofactor expansion),
\begin{eqnarray}
\det \mathrm{M} &=&\sum_{k=1}^{N}\left( -1\right) ^{q_{\pi }+k}\rho \left(
\mathbb{O}_{\pi _{k,N+l}}\left( a^{+}(\varphi _{k}),a(\varphi _{N+q_{\pi
}})\right) \right) \mathrm{M}^{(k,q_{\pi })}  \label{laplace1} \\
\det \mathrm{\tilde{M}} &=&\sum_{k=1}^{N}\left( -1\right) ^{q_{\pi }+k}\rho
\left( a^{+}(\varphi _{k})a(\varphi _{N+q_{\pi }})\right) \mathrm{M}%
^{(k,q_{\pi })}\ .  \label{laplace2}
\end{eqnarray}%
To derive the equality (\ref{laplace2}) we also use that $\tilde{\pi}%
_{k,N+l}=\pi _{k,N+l}$ whenever $l\neq q_{\pi }$, whereas it is the identity
of the set $\{1,2\}$ for $l=q_{\pi }$. On the other hand, using (\ref%
{inductin hypothesis cool1})--(\ref{inductin hypothesis cool2}) and the
induction hypothesis for all $\tilde{N}_{\pi }\geq 0$ with $\tilde{N}_{\pi
}<N_{\pi }$, we deduce that%
\begin{multline*}
\mathrm{M}^{(k,q_{\pi })}=\rho
\Big(%
\mathbb{O}_{\pi ^{(k)}}%
\big(%
a^{+}(\varphi _{1}),\ldots ,a^{+}(\varphi _{k-1}),a^{+}(\varphi
_{k+1}),\ldots ,a^{+}(\varphi _{N}),a(\varphi _{2N}), \\
\ldots ,a(\varphi _{N+q_{\pi }+1}),a(\varphi _{N+q_{\pi }-1}),\ldots
,a(\varphi _{N+1})%
\big)%
\Big)%
\end{multline*}%
and
\begin{equation*}
\det \mathrm{\tilde{M}}=\rho
\Big(%
\mathbb{O}_{\tilde{\pi}}\left( a^{+}(\varphi _{1}),\ldots ,a^{+}(\varphi
_{N}),a(\varphi _{2N}),\ldots ,a(\varphi _{N+1})\right)
\Big)%
\ .
\end{equation*}%
Thus, by induction, it follows from (\ref{CAR AA*}), (\ref{super eqsuper
eqsuper eqsuper eq}), (\ref{laplace1}) and (\ref{laplace2}) that%
\begin{eqnarray*}
&&\rho
\Big(%
\mathbb{O}_{\pi }\left( a^{+}(\varphi _{1}),\ldots ,a^{+}(\varphi
_{N}),a(\varphi _{2N}),\ldots ,a(\varphi _{N+1})\right)
\Big)
\\
&=&\sum_{k=1}^{N}\left( -1\right) ^{q_{\pi }+k}\left\{ \rho \left(
a^{+}(\varphi _{k})a(\varphi _{N+q_{\pi }})\right) -\mathbf{1}\left[ \pi
(k)>k_{\pi }\right] \left\langle \varphi _{N+q_{\pi }},\varphi
_{k}\right\rangle _{\mathcal{H}}\right\} \mathrm{M}^{(k,q_{\pi })} \\
&=&\sum_{k=1}^{N}\left( -1\right) ^{q_{\pi }+k}\rho
\Big(%
\mathbb{O}_{\pi _{k,N+q_{\pi }}}\left( a^{+}(\varphi _{k}),a(\varphi
_{N+q_{\pi }})\right)
\Big)%
\mathrm{M}^{(k,q_{\pi })}=\det \mathrm{M}\ .
\end{eqnarray*}
\end{proof}

\subsection{Representation of Discrete--time Covariance by Quasi--Free
States \label{representation quasi free}}

\noindent \underline{(i):} We pick a (possibly unbounded) self--adjoint
operator $H=H^{\ast }$ acting on $\mathfrak{h}$ and fix from now on $n\in 2%
\mathbb{N}$. Then, because of (\ref{Integral direct C}), (\ref{jean sans
bras0}), (\ref{g explic1bis}) and (\ref{g explic2}), for any fixed $\beta
,\nu \in \mathbb{R}^{+}$\ we introduce the unitary operator%
\begin{equation}
\mathrm{E}\doteq \mathrm{sgn}\left( \mathbf{1}_{\mathfrak{h}}-n^{-1}\beta
H\right) \in \mathcal{B}\left( \mathfrak{h}\right)  \label{En}
\end{equation}%
and the (possibly unbounded) operator $H_{\nu }\doteq \digamma _{\nu }\left(
H\right) $, see (\ref{function F}) and (\ref{Hn0Hn0}). For any $\nu \in
\mathbb{R}$, the Hamiltonian $H_{\nu }$ gives rise to the symbol $S_{\nu }$ (%
\ref{symbol}), which, as explained in\ Section \ref{Sect Tree--Expansions},
in turn yields a quasi--free state $\rho _{S_{\nu }}$, with symbol $S_{\nu
}>0$, on the CAR $C^{\ast }$--algebra $\mathrm{CAR}(\mathfrak{h}\otimes
\mathbb{M})$.

Let
\begin{equation}
\mathfrak{D}\doteq \bigcup_{D\in \mathbb{R}^{+}}\mathrm{ran}\left( \mathbf{1}%
\left[ -D\leq H_{\nu }\leq D\right] \right) \ ,\qquad \nu \in \mathbb{R}\ .
\label{fract domain}
\end{equation}%
By the spectral theorem, it is a dense subspace of entire analytic vectors
of $H_{\nu }$. Note additionally that $\mathfrak{D}$ does not depend on $\nu
\in \mathbb{R}$.\medskip

\noindent \underline{(ii):} Similar to the permutation (\ref{Opi2}), for all
$\alpha _{1},\alpha _{2}\in \mathbb{T}_{n}\cap \lbrack 0,\beta )$, we define
the permutation%
\begin{equation*}
\pi _{\alpha _{1},\alpha _{2}}:\{1,2\}\rightarrow \{1,2\}
\end{equation*}%
as the identity map if $\alpha _{1}\leq \alpha _{2}$, while $\pi _{\alpha
_{1},\alpha _{2}}$ interchanges $1$ and $2$ when $\alpha _{1}>\alpha _{2}$.
\medskip

\noindent Then, quasi--free states $\rho _{S_{\nu }}$, $\nu \in \mathbb{R}%
^{+}$, give rise to the following representation of the discrete--time
covariance:

\begin{lemma}[Representation of the covariance by a quasi--free state]
\label{lemma exp copy(1)}\mbox{ }\newline
Let $\mathfrak{h}$ be any separable Hilbert space. Fix $\beta \in \mathbb{R}%
^{+}$, a self--adjoint operator $H=H^{\ast }$ acting on $\mathfrak{h}$, and $%
n\in 2\mathbb{N}$. Then, for each bounded measurable positive function $%
\kappa $ from $\mathbb{R}$ to $\mathbb{R}_{0}^{+}$, all $m\in \mathbb{N}$,
non--vanishing $\mathfrak{M}\in \mathrm{Mat}\left( m,\mathbb{R}\right) $
with $\mathfrak{M}\geq 0$, $\alpha _{1},\alpha _{2}\in \mathbb{T}_{n}\cap
\lbrack 0,\beta )$, $\varphi _{1},\varphi _{2}\in \mathfrak{D}$ and $%
j_{1},j_{2}\in \{1,\ldots ,m\}$,%
\begin{multline*}
\mathfrak{M}_{j_{1},j_{2}}\left\langle \varphi _{2},\left( C_{H}\kappa (\hat{%
H})\hat{\varphi}_{1}\right) \left( \alpha _{1}-\alpha _{2}\right)
\right\rangle _{\mathfrak{h}} \\
=\lim_{\nu \rightarrow \infty }\rho _{S_{\nu }}%
\Big(%
\mathbb{O}_{\pi _{\alpha _{1},\alpha _{2}}}%
\Big(%
a^{+}\left( (\mathrm{e}^{-\alpha _{1}H_{\nu }}\mathrm{E}^{-\beta
^{-1}n\alpha _{1}}\kappa (H)^{1/2}\varphi _{1})\otimes \mathfrak{e}%
_{j_{1}}\right) , \\
a\left( (\mathrm{e}^{(\alpha _{2}+n^{-1}\beta )H_{\nu }}\mathrm{E}^{\beta
^{-1}n\alpha _{2}+1}\kappa (H)^{1/2}\varphi _{2})\otimes \mathfrak{e}%
_{j_{2}}\right)
\Big)%
\Big)%
\end{multline*}%
with $\mathfrak{e}_{j}\doteq \left[ e_{j}\right] \in \mathbb{M}$ being the
vectors of $\mathbb{M}$ satisfying (\ref{ej fract}) and where $\mathbb{O}%
_{\pi _{\alpha _{1},\alpha _{2}}}$ is defined by (\ref{Opi1}) for $\pi =\pi
_{\alpha _{1},\alpha _{2}}$.
\end{lemma}

\begin{proof}
Fix all the parameters of the lemma. Note that
\begin{equation*}
\left\langle \varphi _{2},\left( C_{H}\kappa (\hat{H})\hat{\varphi}%
_{1}\right) \left( \alpha _{1}-\alpha _{2}\right) \right\rangle _{\mathfrak{h%
}}=\left\langle \kappa (H)^{1/2}\varphi _{2},\left( C_{H}\kappa (\hat{H}%
)^{1/2}\hat{\varphi}_{1}\right) \left( \alpha _{1}-\alpha _{2}\right)
\right\rangle _{\mathfrak{h}}\ .
\end{equation*}%
Therefore, we can assume without loss of generality that $\kappa =\mathbf{1}%
_{\mathbb{R}}$. We deduce from Equations (\ref{Integral direct C}) and (\ref%
{jean sans bras0}) that, for any
\begin{eqnarray}
&&\left\langle \varphi _{2},\left( C_{H}\hat{\varphi}_{1}\right) \left(
\alpha _{1}-\alpha _{2}\right) \right\rangle _{\mathfrak{h}}
\label{eq bonheur1} \\
&=&\left\langle \psi _{2},\left( \int_{\Omega _{H}}^{\oplus }g_{\lambda
_{H}\left( \mathfrak{a}\right) }\ast \hat{\psi}_{1}\left( \mathfrak{a}%
\right) \ \mu _{H}\left( \mathrm{d}\mathfrak{a}\right) \right) \left( \alpha
_{1}-\alpha _{2}\right) \right\rangle _{L^{2}(\Omega _{H};\mathbb{C})}
\notag
\end{eqnarray}%
with
\begin{equation*}
\psi _{1,2}\doteq U_{H}\varphi _{1,2}\qquad \text{and}\qquad \hat{\psi}%
_{1,2}\doteq \hat{U}_{H}\hat{\varphi}_{1,2}\in \ell _{\mathrm{ap}}^{2}(%
\mathbb{T}_{n};L^{2}(\Omega _{H}))\ .
\end{equation*}

In the right--hand side of (\ref{eq bonheur1}) observe that $\hat{\psi}_{1}$
is seen as an element of $L^{2}(\Omega _{H};\ell _{\mathrm{ap}}^{2}(\mathbb{T%
}_{n};\mathbb{C}))$, see (\ref{space cool}). By (\ref{define deltabis}),
observe that
\begin{equation*}
\lbrack \hat{\psi}_{1}(\mathfrak{a})](\alpha )=\delta _{\mathrm{ap}}(\alpha
)\cdot \left( U_{H}\varphi _{1}\right) (\mathfrak{a})\ ,\qquad \mathfrak{a}%
\in \Omega _{H},\ \alpha \in \mathbb{T}_{n}\ ,
\end{equation*}%
which, combined with Equations (\ref{convolution unit}) and (\ref{eq
bonheur1}), yields
\begin{eqnarray}
&&\left\langle \varphi _{2},\left( C_{H}\hat{\varphi}_{1}\right) \left(
\alpha _{1}-\alpha _{2}\right) \right\rangle _{\mathfrak{h}}
\label{eq cool 00} \\
&=&\left\langle \psi _{2},\left( \int_{\Omega _{H}}^{\oplus }g_{\lambda
_{H}\left( \mathfrak{a}\right) }\left( \alpha _{1}-\alpha _{2}\right) \psi
_{1}\left( \mathfrak{a}\right) \mu _{H}\left( \mathrm{d}\mathfrak{a}\right)
\right) \right\rangle _{L^{2}(\Omega _{H};\mathbb{C})}.  \notag
\end{eqnarray}%
Therefore, by using the explicit expressions (\ref{Hn0Hn0}), (\ref{g
explic1bis})--(\ref{g explic2}) and (\ref{En}), we deduce from (\ref%
{spectral theorem}) and (\ref{eq cool 00}) the equality
\begin{equation}
\left\langle \varphi _{2},\left( C_{H}\hat{\varphi}_{1}\right) \left( \alpha
_{1}-\alpha _{2}\right) \right\rangle _{\mathfrak{h}}=\lim_{\nu \rightarrow
\infty }\left\langle \mathrm{E}^{\beta ^{-1}n\alpha _{2}+1}\varphi _{2},%
\frac{\mathrm{e}^{(\alpha _{2}+n^{-1}\beta -\alpha _{1})H_{\nu }}}{1+\mathrm{%
e}^{\beta H_{\nu }}}\mathrm{E}^{-\beta ^{-1}n\alpha _{1}}\varphi
_{1}\right\rangle _{\mathfrak{h}}  \label{eq fin1}
\end{equation}%
for any $\alpha _{1}\leq \alpha _{2}$ while, for any $\alpha _{1}>\alpha
_{2} $,
\begin{eqnarray}
&&\left\langle \varphi _{2},\left( C_{H}\hat{\varphi}_{1}\right) \left(
\alpha _{1}-\alpha _{2}\right) \right\rangle _{\mathfrak{h}}  \label{eq fin2}
\\
&=&-\lim_{\nu \rightarrow \infty }\left\langle \mathrm{E}^{\beta
^{-1}n\alpha _{2}+1}\varphi _{2},\frac{\mathrm{e}^{(\beta -(\alpha
_{1}-\alpha _{2}-n^{-1}\beta ))H_{\nu }}}{1+\mathrm{e}^{\beta H_{\nu }}}%
\mathrm{E}^{-\beta ^{-1}n\alpha _{1}}\varphi _{1}\right\rangle _{\mathfrak{h}%
}\ ,  \notag
\end{eqnarray}%
using $n\in 2\mathbb{N}$. On the other hand, if $\alpha _{1}\leq \alpha _{2}$
then $\pi _{\alpha _{1},\alpha _{2}}=\mathbf{1}_{\{1,2\}}$ and we infer from
Equations (\ref{scalar product}), (\ref{ej fract}), (\ref{quasi free symbol}%
), (\ref{symbol}) and (\ref{Opi1}) that
\begin{eqnarray}
&&\rho _{S_{\nu }}%
\Big(%
\mathbb{O}_{\pi _{\alpha _{1},\alpha _{2}}}%
\Big(%
a^{+}\left( (\mathrm{e}^{-\alpha _{1}H_{\nu }}\mathrm{E}^{-\beta
^{-1}n\alpha _{1}}\varphi _{1})\otimes \mathfrak{e}_{j_{1}}\right) ,  \notag
\\
&&\qquad \qquad \qquad \qquad a\left( (\mathrm{e}^{(\alpha _{2}+n^{-1}\beta
)H_{\nu }}\mathrm{E}^{\beta ^{-1}n\alpha _{2}+1}\varphi _{2})\otimes
\mathfrak{e}_{j_{2}}\right)
\Big)%
\Big)
\notag \\
&=&\mathfrak{M}_{j_{1},j_{2}}\left\langle \mathrm{E}^{\beta ^{-1}n\alpha
_{2}+1}\varphi _{2},\frac{\mathrm{e}^{(\alpha _{2}+n^{-1}\beta -\alpha
_{1})H_{\nu }}}{1+\mathrm{e}^{\beta H_{\nu }}}\mathrm{E}^{-\beta
^{-1}n\alpha _{1}}\varphi _{1}\right\rangle _{\mathfrak{h}}  \label{eq fin3}
\end{eqnarray}%
and the assertion holds true when $\alpha _{1}\leq \alpha _{2}$. If $\alpha
_{1}>\alpha _{2}$ then
\begin{multline*}
\rho _{S_{\nu }}%
\Big(%
\mathbb{O}_{\pi _{\alpha _{1},\alpha _{2}}}%
\Big(%
a^{+}\left( (\mathrm{e}^{-\alpha _{1}H_{\nu }}\mathrm{E}^{-\beta
^{-1}n\alpha _{1}}\varphi _{1})\otimes \mathfrak{e}_{j_{1}}\right) , \\
a\left( (\mathrm{e}^{(\alpha _{2}+n^{-1}\beta )H_{\nu }}\mathrm{E}^{\beta
^{-1}n\alpha _{2}+1}\varphi _{2})\otimes \mathfrak{e}_{j_{2}}\right)
\Big)%
\Big)
\\
=-\rho _{S_{\nu }}%
\Big(%
a\left( (\mathrm{e}^{(\alpha _{2}+n^{-1}\beta )H_{\nu }}\mathrm{E}^{\beta
^{-1}n\alpha _{2}+1}\varphi _{2})\otimes \mathfrak{e}_{j_{2}}\right) \\
a^{+}\left( (\mathrm{e}^{-\alpha _{1}H_{\nu }}\mathrm{E}^{-\beta
^{-1}n\alpha _{1}}\varphi _{1})\otimes \mathfrak{e}_{j_{1}}\right)
\Big)%
\ ,
\end{multline*}%
which, combined with (\ref{CAR AA*}), implies that%
\begin{multline*}
\rho _{S_{\nu }}%
\Big(%
\mathbb{O}_{\pi _{\alpha _{1},\alpha _{2}}}%
\Big(%
a^{+}\left( (\mathrm{e}^{-\alpha _{1}H_{\nu }}\mathrm{E}^{-\beta
^{-1}n\alpha _{1}}\varphi _{1})\otimes \mathfrak{e}_{j_{1}}\right) , \\
a\left( (\mathrm{e}^{(\alpha _{2}+n^{-1}\beta )H_{\nu }}\mathrm{E}^{\beta
^{-1}n\alpha _{2}+1}\varphi _{2})\otimes \mathfrak{e}_{j_{2}}\right)
\Big)%
\Big)
\\
=\rho _{S_{\nu }}%
\Big(%
a^{+}\left( (\mathrm{e}^{-\alpha _{1}H_{\nu }}\mathrm{E}^{-\beta
^{-1}n\alpha _{1}}\varphi _{1})\otimes \mathfrak{e}_{j_{1}}\right) \\
\qquad \qquad \qquad \qquad \qquad \qquad \qquad \qquad a\left( (\mathrm{e}%
^{(\alpha _{2}+n^{-1}\beta )H_{\nu }}\mathrm{E}^{\beta ^{-1}n\alpha
_{2}+1}\varphi _{2})\otimes \mathfrak{e}_{j_{2}}\right)
\Big)
\\
-\left\langle (\mathrm{e}^{(\alpha _{2}+n^{-1}\beta )H_{\nu }}\mathrm{E}%
^{\beta ^{-1}n\alpha _{2}+1}\varphi _{2})\otimes \mathfrak{e}_{j_{2}},(%
\mathrm{e}^{-\alpha _{1}H_{\nu }}\mathrm{E}^{-\beta ^{-1}n\alpha
_{1}}\varphi _{1})\otimes \mathfrak{e}_{j_{1}}\right\rangle _{\mathfrak{h}%
\otimes \mathbb{M}}\ .
\end{multline*}%
Using again (\ref{scalar product}), (\ref{ej fract}) and (\ref{symbol}), we
thus arrive from the last equality at
\begin{eqnarray}
&&\rho _{S_{\nu }}%
\Big(%
\mathbb{O}_{\pi _{\alpha _{1},\alpha _{2}}}%
\Big(%
a^{+}\left( (\mathrm{e}^{-\alpha _{1}H_{\nu }}\mathrm{E}^{-\beta
^{-1}n\alpha _{1}}\varphi _{1})\otimes \mathfrak{e}_{j_{1}}\right) ,  \notag
\\
&&\qquad \qquad \qquad \qquad a\left( (\mathrm{e}^{(\alpha _{2}+n^{-1}\beta
)H_{\nu }}\mathrm{E}^{\beta ^{-1}n\alpha _{2}+1}\varphi _{2})\otimes
\mathfrak{e}_{j_{2}}\right)
\Big)%
\Big)
\notag \\
&=&-\mathfrak{M}_{j_{1},j_{2}}\left\langle \mathrm{E}^{\beta ^{-1}n\alpha
_{2}+1}\varphi _{2},\frac{\mathrm{e}^{(\beta -(\alpha _{1}-\alpha
_{2}-n^{-1}\beta ))H_{\nu }}}{1+\mathrm{e}^{\beta H_{\nu }}}\mathrm{E}%
^{-\beta ^{-1}n\alpha _{1}}\varphi _{1}\right\rangle _{\mathfrak{h}}\ .
\label{eq fin4}
\end{eqnarray}%
By combining (\ref{eq fin1}) and (\ref{eq fin2}) with (\ref{eq fin3}) and (%
\ref{eq fin4}), we arrive at the assertion with $\kappa =\mathbf{1}_{\mathbb{%
R}}$.
\end{proof}

\begin{koro}[Determinants of the covariance and quasi--free states]
\label{Corollary chiot1}\mbox{
}\newline
Let $\mathfrak{h}$ be any separable Hilbert space. Fix $\beta \in \mathbb{R}%
^{+}$, a self--adjoint operator $H=H^{\ast }$ acting on $\mathfrak{h}$, and $%
n\in 2\mathbb{N}$. Then, for each bounded measurable positive function $%
\kappa $ from $\mathbb{R}$ to $\mathbb{R}_{0}^{+}$, all $m,N\in \mathbb{N}$,
non--vanishing $\mathfrak{M}\in \mathrm{Mat}\left( m,\mathbb{R}\right) $
with $\mathfrak{M}\geq 0$\ and
\begin{equation*}
\{(\alpha _{q},\varphi _{q},j_{q})\}_{q=1}^{2N}\subset \mathbb{T}_{n}\cap
\lbrack 0,\beta )\times \mathfrak{D}\times \{1,\ldots ,m\}\ ,
\end{equation*}%
the following identity holds true:%
\begin{align}
& \mathrm{det}\left[ \mathfrak{M}_{j_{k},j_{N+l}}\left\langle \varphi
_{N+l},\left( C_{H}\kappa (\hat{H})\hat{\varphi}_{k}\right) \left( \alpha
_{k}-\alpha _{N+l}\right) \right\rangle _{\mathfrak{h}}\right] _{k,l=1}^{N}
\notag \\
& =\lim_{\nu \rightarrow \infty }\rho _{S_{\nu }}%
\Big(%
\mathbb{O}_{\pi }%
\Big(%
a^{+}\left( (\mathrm{e}^{-\tilde{\alpha}_{1}H_{\nu }}\mathrm{E}^{-\beta
^{-1}n\tilde{\alpha}_{1}}\kappa (H)^{1/2}\varphi _{1})\otimes \mathfrak{e}%
_{j_{1}}\right) ,\ldots ,  \label{eq horloge} \\
& \qquad \ \qquad \ \qquad \ a^{+}\left( (\mathrm{e}^{-\tilde{\alpha}%
_{N}H_{\nu }}\mathrm{E}^{-\beta ^{-1}n\tilde{\alpha}_{N}}\kappa
(H)^{1/2}\varphi _{N})\otimes \mathfrak{e}_{j_{N}}\right) ,  \notag \\
& \qquad \ \qquad \ \qquad \ \qquad \ a\left( (\mathrm{e}^{\tilde{\alpha}%
_{2N}H_{\nu }}\mathrm{E}^{\beta ^{-1}n\tilde{\alpha}_{2N}}\kappa
(H)^{1/2}\varphi _{2N})\otimes \mathfrak{e}_{j_{2N}}\right) ,  \notag \\
& \qquad \qquad \qquad \qquad \quad \ \ldots ,a\left( (\mathrm{e}^{\tilde{%
\alpha}_{N+1}H_{\nu }}\mathrm{E}^{\beta ^{-1}n\tilde{\alpha}_{N+1}}\kappa
(H)^{1/2}\varphi _{N+1})\otimes \mathfrak{e}_{j_{N+1}}\right)
\Big)%
\Big)
\notag
\end{align}%
for any permutation $\pi $ of $2N$ elements such that
\begin{equation}
\tilde{\alpha}_{\pi ^{-1}(q)}-\tilde{\alpha}_{\pi ^{-1}(q-1)}\geq 0\ ,\text{
\ }\alpha _{\pi ^{-1}(q)}-\alpha _{\pi ^{-1}(q-1)}\geq 0\ ,\quad q\in
\left\{ 2,\ldots ,2N\right\} \ ,  \label{ordering}
\end{equation}%
where $\tilde{\alpha}_{q}\doteq \alpha _{q}$ for $q\in \{1,\ldots ,N\}$ and $%
\tilde{\alpha}_{q}\doteq \alpha _{q}+n^{-1}\beta $ for $q\in \{N+1,\ldots
,2N\}$.
\end{koro}

\begin{proof}
Fix all the parameters of the corollary. Take any permutation $\pi $ of $2N$
elements such that%
\begin{equation}
\pi _{\alpha _{k},\alpha _{N+l}}=\pi _{k,N+l}\ ,\qquad k,l\in \{1,\ldots
,N\}\ .  \label{orderingordering}
\end{equation}%
See, respectively, (iv) before Lemma \ref{lemma exp copy(1)} and Equation (%
\ref{Opi2}) for the definitions of the permutations $\pi _{\alpha
_{k},\alpha _{N+l}}$ and $\pi _{k,N+l}$ of two elements. Then, (\ref{eq
horloge}) follows from Lemmata \ref{lemma exp copy(2)} and \ref{lemma exp
copy(1)}. To conclude the proof observe that a permutation $\pi $ of $2N$
elements satisfying (\ref{ordering}) exists and also satisfies (\ref%
{orderingordering}), keeping in mind Equation (\ref{definition torus}).
\end{proof}

\subsection{Correlation Functions and Tomita--Takesaki Modular Theory\label%
{Sectino chiot}}

\noindent \underline{(i):} As above, fix $\beta \in \mathbb{R}^{+}$, a
self--adjoint operator $H=H^{\ast }$ acting on $\mathfrak{h}$, $n\in 2%
\mathbb{N}$, $\nu \in \mathbb{R}^{+}$, and a non--vanishing positive real
matrix $\mathfrak{M}\in \mathrm{Mat}\left( m,\mathbb{R}\right) $ with $m\in
\mathbb{N}$. Let $\tau \equiv \{\tau _{t}\}_{t\in {\mathbb{R}}}$ be the
unique $C_{0}$--group (that is, strongly continuous group) of auto%
\-%
morphisms on the $C^{\ast }$--algebra $\mathrm{CAR}(\mathfrak{h}\otimes
\mathbb{M})$ satisfying%
\begin{equation}
\tau _{t}\left( a\left( \varphi \otimes g\right) \right) =a(\mathrm{e}%
^{itH_{\nu }\otimes \mathbf{1}_{\mathfrak{h}\otimes \mathbb{M}}}\varphi
\otimes g)=a\left( (\mathrm{e}^{itH_{\nu }}\varphi )\otimes g\right) \
,\qquad \varphi \in \mathfrak{h},\ g\in \mathbb{M}\ .  \label{definition tho}
\end{equation}%
See (\ref{Hn0Hn0}). It is well---known that the quasi--free state $\rho
_{S_{\nu }}$, which is defined from the symbol $S_{\nu }$ (\ref{symbol}), is
the unique $(\tau ,\beta )$--KMS state on $\mathrm{CAR}(\mathfrak{h}\otimes
\mathbb{M})$. \medskip

\noindent \underline{(ii):} Recall that $(\mathfrak{H}_{\nu },\varkappa
_{\nu },\eta _{\nu })$ is a cyclic\ representation of $\rho _{S_{\nu }}$
(Section \ref{Modular}). The weak closure of the $C^{\ast }$--algebra $%
\mathrm{CAR}(\mathfrak{h}\otimes \mathbb{M})$ is the von Neumann algebra $%
\mathcal{X}_{\nu }$ (\ref{sdfjkl}). The state $\rho _{S_{\nu }}\circ
\varkappa _{\nu }$ on $\varkappa _{\nu }(\mathrm{CAR}(\mathfrak{h}\otimes
\mathbb{M}))$ extends uniquely to a normal state on the von Neumann algebra $%
\mathcal{X}_{\nu }$ and the $C_{0}$--group $\{\tau _{t}\circ \varkappa _{\nu
}\}_{t\in \mathbb{R}}$ also uniquely extends to a $\sigma $--weakly
continuous $\ast $--automorphism group on $\mathcal{X}_{\nu }$. Both
extensions are again denoted by $\rho _{S_{\nu }}$ and $\{\tau _{t}\}_{t\in
\mathbb{R}}$, respectively. By \cite[Corollary 5.3.4]{BratteliRobinson}, $%
\rho _{S_{\nu }}$ is again a $(\tau ,\beta )$--KMS state on $\mathcal{X}%
_{\nu }$. \medskip

\noindent \underline{(iii):} By \cite[Corollary 5.3.9]{BratteliRobinson},
the cyclic vector $\eta _{\nu }$\ is separating for $\mathcal{X}_{\nu }$,
i.e., $A\eta _{\nu }=0$ implies $A=0$ for all $A\in \mathcal{X}_{\nu }$.
Denote by $\Delta _{\nu }$ the (possibly unbounded) Tomita--Takesaki modular
operator of the pair $\left( \mathcal{X}_{\nu },\eta _{\nu }\right) $. The ($%
\beta $--rescaled) modular group is the $\sigma $--weakly continuous
one-parameter group $\sigma \equiv \{\sigma _{t}\}_{t\in {\mathbb{R}}}$
defined by%
\begin{equation}
\sigma _{t}\left( A\right) \doteq \Delta _{\nu }^{-it\beta ^{-1}}A\Delta
_{\nu }^{it\beta ^{-1}}\qquad A\in \mathcal{X}_{\nu }\ .
\label{definition thobis}
\end{equation}%
(If $\beta =-1$ then $\sigma $ is the well--known modular automorphism group
associated with the pair $\left( \mathcal{X}_{\nu },\eta _{\nu }\right) $,
see \cite[Definition 2.5.15]{BratteliRobinsonI}.) By Takesaki's theorem \cite%
[Theorem 5.3.10]{BratteliRobinson}, we deduce that $\sigma =\tau $. In
particular, using (\ref{definition tho}) we arrive at the equality
\begin{equation}
\varkappa _{\nu }\left( a\left( (\mathrm{e}^{itH_{\nu }}\varphi )\otimes
g\right) \right) =\Delta _{\nu }^{-it\beta ^{-1}}\varkappa _{\nu }\left(
a\left( \varphi \otimes g\right) \right) \Delta _{\nu }^{it\beta ^{-1}}\
,\qquad \varphi \in \mathfrak{h},\ g\in \mathbb{M}\ .
\label{assertion chiot bresiliennes}
\end{equation}%
\medskip \noindent \underline{(iv):} Recall that $\mathfrak{D}\subseteq
\mathfrak{h}$ (\ref{fract domain}) is a dense subspace of entire analytic
vectors for $H_{\nu }$, while for any $N\in \mathbb{N}$ and $\zeta \in
\mathbb{R}^{+}$, $\mathfrak{T}_{N}^{(\zeta )}$ is the tube defined by (\ref%
{tube}). For any $\varphi \in \mathfrak{D}$ and $g\in \mathbb{M}$, the maps
\begin{equation}
z\mapsto a^{+}\left( (\mathrm{e}^{-zH_{\nu }}\varphi )\otimes g\right)
\qquad \text{and}\qquad z\mapsto a\left( (\mathrm{e}^{\bar{z}H_{\nu
}}\varphi )\otimes g\right)  \label{map tricial}
\end{equation}%
from $\mathbb{C}$ to the $C^{\ast }$--algebra $\mathrm{CAR}(\mathfrak{h}%
\otimes \mathbb{M})$ are entire analytic functions. Fix $N\in \mathbb{N}$, $%
\varphi _{1},\ldots ,\varphi _{N}\in \mathfrak{D}$, $g_{1},\ldots ,g_{N}\in
\mathbb{M}$, and pick a family
\begin{equation}
\left\{ a^{\#}\left( \varphi _{q}\otimes g_{q}\right) \right\}
_{q=1}^{N}\subset \mathrm{CAR}\left( \mathfrak{h}\otimes \mathbb{M}\right) \
,  \label{family}
\end{equation}%
where the notation \textquotedblleft $a^{\#}$\textquotedblright\ stands for
either \textquotedblleft $a^{+}$\textquotedblright\ or \textquotedblleft $a$%
\textquotedblright . For any $\varphi \in \mathfrak{D}$, $g\in \mathbb{M}$
and $z\in \mathbb{C}$, we also use the convention
\begin{equation*}
a^{\#}\left( (\mathrm{e}^{z^{\#}H_{\nu }}\varphi )\otimes g\right) =\left\{
\begin{array}{lll}
a^{+}\left( (\mathrm{e}^{-zH_{\nu }}\varphi )\otimes g\right) & \text{when}
& a^{\#}=a^{+}, \\
a\left( (\mathrm{e}^{\bar{z}H_{\nu }}\varphi )\otimes g\right) & \text{when}
& a^{\#}=a\ ,%
\end{array}%
\right.
\end{equation*}%
with%
\begin{equation*}
(z_{1}+z_{2})^{\#}\doteq z_{1}^{\#}+z_{2}^{\#}\ ,\qquad z_{1},z_{2}\in
\mathbb{C}\ .
\end{equation*}%
Then, for any fixed integer $p\in \{1,\ldots ,N\}$, the map $\Upsilon $ from
$\mathbb{C}^{N+1}$ to $\mathbb{C}$ defined by%
\begin{align}
& \Upsilon (z_{1},\ldots ,z_{p},\tilde{z}_{p},z_{p+1},\ldots ,z_{N})
\label{gamma} \\
& \doteq \rho
\Big(%
a^{\#}%
\big(%
(\mathrm{e}^{z_{1}^{\#}H_{\nu }}\varphi _{1})\otimes g_{1}%
\big)%
\cdots a^{\#}\left( (\mathrm{e}^{(z_{1}^{\#}+\cdots +z_{p-1}^{\#})H_{\nu
}}\varphi _{p-1})\otimes g_{p-1}\right)  \notag \\
& \qquad \quad a^{\#}%
\big(%
(\mathrm{e}^{(z_{1}^{\#}+\cdots +z_{p}^{\#}+\tilde{z}_{p}^{\#})H_{\nu
}}\varphi _{p})\otimes g_{p}%
\big)%
a^{\#}%
\big(%
(\mathrm{e}^{(z_{1}^{\#}+\cdots +z_{p}^{\#}+\tilde{z}_{p}^{\#}+z_{p+1}^{%
\#})H_{\nu }}\varphi _{p+1})\otimes g_{p+1}%
\big)
\notag \\
& \qquad \qquad \qquad \qquad \qquad \qquad \quad \cdots a^{\#}%
\big(%
(\mathrm{e}^{(z_{1}^{\#}+\cdots +z_{p}^{\#}+\tilde{z}_{p}^{\#}+z_{p+1}^{\#}+%
\cdots +z_{N}^{\#})H_{\nu }}\varphi _{N})\otimes g_{N}%
\big)%
\Big)
\notag
\end{align}%
is an entire analytic function. \medskip

\noindent \underline{(v):} For $\Psi _{1},\ldots ,\Psi _{N}\in \mathfrak{h}%
\otimes \mathbb{M}$, define a family $\{a^{\#}\left( \Psi _{q}\right)
\}_{q=1}^{N}$ of elements of the $C^{\ast }$--algebra $\mathrm{CAR}(%
\mathfrak{h}\otimes \mathbb{M})$, where, as before, \textquotedblleft $%
a^{\#}=a^{+}$\textquotedblright\ or \textquotedblleft $a^{\#}=a$%
\textquotedblright . Then, by applying \cite[Lemma A.1]{Araki-Masuda}, we
obtain the following assertions:

\begin{itemize}
\item[(A)] The (cyclic and separating) vector $\eta _{\nu }$ belongs to the
domain of definition of the possibly unbounded operator
\begin{equation*}
\Delta _{\nu }^{z_{1}\beta ^{-1}}\varkappa _{\nu }\left( a^{\#}\left( \Psi
_{1}\right) \right) \cdots \Delta _{\nu }^{z_{N}\beta ^{-1}}\varkappa _{\nu
}\left( a^{\#}\left( \Psi _{N}\right) \right)
\end{equation*}%
for any $(z_{1},\ldots ,z_{N})\in \mathfrak{T}_{N}^{(\beta /2)}$ with%
\begin{equation}
\left\Vert \Delta _{\nu }^{z_{1}\beta ^{-1}}\varkappa _{\nu }\left(
a^{\#}\left( \Psi _{1}\right) \right) \cdots \Delta _{\nu }^{z_{N}\beta
^{-1}}\varkappa _{\nu }\left( a^{\#}\left( \Psi _{N}\right) \right) \eta
_{\nu }\right\Vert _{\mathfrak{H}_{\nu }}\leq \prod_{q=1}^{N}\left\Vert \Psi
_{q}\right\Vert _{\mathfrak{h}\otimes \mathbb{M}}\ .
\label{property araki chiot 1}
\end{equation}%
This inequality is a special case of \cite[(A.2)]{Araki-Masuda}, which is
intimately related to H\"{o}lder inequalities for non--commutative $L^{p}$%
--spaces.

\item[(B)] The map from $\mathfrak{T}_{N}^{(\beta /2)}$ to $\mathfrak{H}%
_{\nu }$\ defined by
\begin{equation*}
(z_{1},\ldots ,z_{N})\mapsto \Delta _{\nu }^{z_{1}\beta ^{-1}}\varkappa
_{\nu }\left( a^{\#}\left( \Psi _{1}\right) \right) \cdots \Delta _{\nu
}^{z_{N}\beta ^{-1}}\varkappa _{\nu }\left( a^{\#}\left( \Psi _{N}\right)
\right) \eta _{\nu }
\end{equation*}%
is norm continuous on the whole tube $\mathfrak{T}_{N}^{(\beta /2)}$ and
analytic on its interior.
\end{itemize}

\noindent Using the notation
\begin{equation*}
x_{q}\doteq \varkappa _{\nu }\left( a^{\#}(\varphi _{q}\otimes g_{q})\right)
\ ,\qquad q\in \left\{ 1,\ldots ,N\right\} \ ,
\end{equation*}%
we consider now the map $\Theta $ from $\mathfrak{T}_{p}^{(\beta /2)}\times
\mathfrak{T}_{N-p+1}^{(\beta /2)}$ to $\mathbb{R}$ defined by%
\begin{multline*}
\Theta
\big(%
(z_{1},\ldots ,z_{p}),(\tilde{z}_{p},z_{p+1},\ldots ,z_{N})%
\big)
\\
\doteq \left\langle \Delta _{\nu }^{\bar{z}_{p}\beta ^{-1}}x_{p-1}^{\ast
}\Delta _{\nu }^{\bar{z}_{p-1}\beta ^{-1}}\cdots x_{2}^{\ast }\Delta _{\nu
}^{\bar{z}_{2}\beta ^{-1}}x_{1}^{\ast }\eta _{\nu }\ ,\right. \\
\left. \Delta _{\nu }^{\tilde{z}_{p}\beta ^{-1}}x_{p}\Delta _{\nu
}^{z_{p+1}\beta ^{-1}}x_{p+1}\cdots \Delta _{\nu }^{z_{N}\beta
^{-1}}x_{N}\eta _{\nu }\right\rangle _{\mathfrak{H}_{\nu }}.
\end{multline*}%
Compare for instance with \cite[Lemma A]{Araki-Masuda}, which explains the
properties of $\Theta $. (Notice that \cite{Araki-Masuda} uses a different
convention for sesquilinear forms.) By (\ref{property araki chiot 1}), this
function is uniformly bounded for all $n\in 2${$\mathbb{N}$}. The same is
trivially true for the map $\Upsilon $ (\ref{gamma}) on%
\begin{equation*}
\mathfrak{T}_{p}^{(\beta /2)}\times \mathfrak{T}_{N-p+1}^{(\beta /2)}\subset
\mathbb{C}^{N+1}\ .
\end{equation*}%
Moreover, by using (\ref{assertion chiot bresiliennes}) we deduce that $%
\Upsilon $ and $\Theta $ are equal to each other on $i\mathbb{R}^{p}\times i%
\mathbb{R}^{N-p+1}$. For each fixed imaginary vector $(z_{1},\ldots
,z_{p})\in i\mathbb{R}^{p}$, the maps $\Upsilon $ and $\Theta $ are both
continuous as functions of $(\tilde{z}_{p},z_{p+1},\ldots ,z_{N})\in
\mathfrak{T}_{N-p+1}^{(\beta /2)}$ and analytic in the interior of $%
\mathfrak{T}_{N-p+1}^{(\beta /2)}$, by (B) \cite[Lemma A.1]{Araki-Masuda}.
Hence, from Hadamard's three line theorem (see, e.g., \cite[Appendix to IX.4]%
{ReedSimonII}), $\Upsilon $ and $\Theta $ are equal to each other for any
fixed imaginary vector $(z_{1},\ldots ,z_{p})\in i\mathbb{R}^{p}$ and all
complex vectors $(\tilde{z}_{p},z_{p+1},\ldots ,z_{N})\in \mathfrak{T}%
_{N-p+1}^{(\beta /2)}$. Applying this argument again at fixed $(\tilde{z}%
_{p},z_{p+1},\ldots ,z_{N})\in \mathfrak{T}_{N-p+1}^{(\beta /2)}$ for $%
\Upsilon $ and $\Theta $ viewed as functions of $(z_{1},\ldots ,z_{p})\in
\mathfrak{T}_{p}^{(\beta /2)}$, we conclude that $\Upsilon =\Theta $ on $%
\mathfrak{T}_{p}^{(\beta /2)}\times \mathfrak{T}_{N-p+1}^{(\beta /2)}$.

In particular, for $N\in \mathbb{N}$, any family (\ref{family}) of elements
of the $C^{\ast }$--algebra $\mathrm{CAR}(\mathfrak{h}\otimes \mathbb{M})$,
and all $\alpha _{1},\ldots ,\alpha _{N}\in \lbrack 0,\beta ]$ such that%
\begin{equation*}
\vartheta _{q}\doteq \beta ^{-1}(\alpha _{q}-\alpha _{q-1})\geq 0\ ,\qquad
q\in \left\{ 2,\ldots ,N\right\} \ ,
\end{equation*}%
the following equality holds true:
\begin{align}
& \rho
\Big(%
a^{\#}%
\big(%
(\mathrm{e}^{\alpha _{1}^{\#}H_{\nu }}\varphi _{1})\otimes g_{1}%
\big)%
\ldots a^{\#}%
\big(%
(\mathrm{e}^{\alpha _{N}^{\#}H_{\nu }}\varphi _{N})\otimes g_{N}%
\big)%
\Big)
\notag \\
& =\left\langle \Delta _{\nu }^{\frac{1}{2}-\beta ^{-1}\alpha
_{p-1}}x_{p-1}^{\ast }\Delta _{\nu }^{\vartheta _{p-1}}\cdots x_{2}^{\ast
}\Delta _{\nu }^{\vartheta _{2}}x_{1}^{\ast }\eta _{\nu }\ ,\right.
\label{chiot final1} \\
& \qquad \qquad \qquad \qquad \left. \Delta _{\nu }^{\beta ^{-1}\alpha _{p}-%
\frac{1}{2}}x_{p}\Delta _{\nu }^{\vartheta _{p+1}}x_{p+1}\cdots \Delta _{\nu
}^{\vartheta _{N}}x_{N}\eta _{\nu }\right\rangle _{\mathfrak{H}_{\nu }}
\notag
\end{align}%
with $p$ defined to be the smallest element of $\left\{ 1,\ldots ,N\right\} $
such that $\alpha _{p}\geq \beta /2$. \bigskip

\noindent \textit{Acknowledgments:} This research is supported by the\
FAPESP, the CNPq, the Basque Government through the grant IT641-13 and the
BERC 2014-2017 program and by the Spanish Ministry of Economy and
Competitiveness MINECO: BCAM Severo Ochoa accreditation SEV-2013-0323 and
MTM2014-53850. We are very grateful to the BCAM and its management, which
supported this project via the visiting researcher program. Finally, we
thank the referee for his thorough revision work, Zosza Lefevre for
linguistic hints and Christian J\"{a}kel for the nice lectures he gave in
2015 on non--commutative $L^{p}$--spaces.

\end{document}